\documentclass[12pt]{article}

\usepackage{amsmath,amscd,amssymb}
\usepackage{mathtools}
\usepackage{etoolbox}
\usepackage{enumitem}

\usepackage[ngerman,english]{babel}
\usepackage[utf8]{inputenc}
\usepackage[T1]{fontenc}

\usepackage{textgreek}

\usepackage{flafter}

\usepackage{tikz}
\usetikzlibrary{calc,arrows,decorations,decorations.pathreplacing,backgrounds,shapes,fit}
\makeatletter
\newlength{\SK@tmpl}
\tikzset{%
  text width of/.code={\setlength{\SK@tmpl}{\widthof{\csuse{tikz@textfont}#1}}\pgfkeysalso{text width/.expanded=\the\SK@tmpl}},
  text height of/.code={\setlength{\SK@tmpl}{\heightof{\csuse{tikz@textfont}#1}}\pgfkeysalso{text height/.expanded=\the\SK@tmpl}},
  text depth of/.code={\setlength{\SK@tmpl}{\depthof{\csuse{tikz@textfont}#1}}\pgfkeysalso{text depth/.expanded=\the\SK@tmpl}},
  minimum height of/.code={\setlength{\SK@tmpl}{\totalheightof{\csuse{tikz@textfont}#1}}\pgfkeysalso{minimum height/.expanded=\the\SK@tmpl}},
  standard height/.style={text height of=M,text depth of=p,text width/.append style=natural height},
  natural height/.style={text height=,text depth=}}

\pgfarrowsdeclare{interval}{interval}
{
  \pgfarrowsleftextend{+-0.5\pgflinewidth}
  \pgfarrowsrightextend{+.5\pgflinewidth}
}
{
  \pgfutil@tempdima=1pt%
  \advance\pgfutil@tempdima by1.25\pgflinewidth%
  \pgfsetdash{}{+0pt}
  \pgfsetrectcap
  \pgfpathmoveto{\pgfqpoint{0pt}{-\pgfutil@tempdima}}
  \pgfpathlineto{\pgfqpoint{0pt}{\pgfutil@tempdima}}
  \pgfusepathqstroke
}

\newif\ifdrawcapoints
\tikzset{%
  ca/.is family,
  ca/.search also={/tikz},
  ca,
  cabasestyle/.style=,
  cabase/.initial=1cm,
  castep/.initial=.2cm,
  caanglestep/.initial=1,
  calevelsx/.initial=0,
  calevelsy/.initial=0,
  calevels/.style={calevelsx=#1,calevelsy=#1},
  at calevels/.style={},
  capoints/.is if=drawcapoints,
  every carc/.style={interval},
  start/.initial={},
  end/.initial={},
  startlevel/.initial=1,
  endlevel/.initial=1,
  level/.style={startlevel=#1,endlevel=#1},
  label/.initial={},
  labelpos/.initial=0.5,
  nodepos/.style={labelpos=#1},
  labelangle/.initial={},
  nodeangle/.style={labelangle=#1},
  every label/.style={inner sep=2pt},
  startlabel/.initial={},
  startlabelsep/.initial=\CAcirclesep,
  startlabelanchoroffset/.initial=-90,
  startlabelpos/.is choice,
  startlabelpos/cw/.style={startlabelanchoroffset=+90,startlabelsep=\CAcirclesep,startlabelposauto=false},
  startlabelpos/ccw/.style={startlabelanchoroffset=-90,startlabelsep=\CAcirclesep,startlabelposauto=false},
  startlabelpos/outside/.style={startlabelanchoroffset=+180,startlabelsep=1pt+1.75\pgflinewidth,startlabelposauto=false},
  startlabelpos/inside/.style={startlabelanchoroffset=,startlabelsep=1pt+1.75\pgflinewidth,startlabelposauto=false},
  startlabelposauto/.initial=true,
  every startlabel/.style={inner sep=2pt},
  endlabel/.initial={},
  endlabelsep/.initial=\CAcirclesep,
  endlabelanchoroffset/.initial=+90,
  endlabelpos/.is choice,
  endlabelpos/cw/.style={endlabelanchoroffset=+90,endlabelsep=\CAcirclesep,endlabelposauto=false},
  endlabelpos/ccw/.style={endlabelanchoroffset=-90,endlabelsep=\CAcirclesep,endlabelposauto=false},
  endlabelpos/outside/.style={endlabelanchoroffset=+180,endlabelsep=1pt+1.75\pgflinewidth,endlabelposauto=false},
  endlabelpos/inside/.style={endlabelanchoroffset=,endlabelsep=1pt+1.75\pgflinewidth,endlabelposauto=false},
  endlabelposauto/.initial=true,
  every endlabel/.style={inner sep=2pt},
}
\newenvironment{camodel}[1][]{
  \begin{scope}[ca,#1]
    \path[ca,cabasestyle] (0,0) circle(\pgfkeysvalueof{/tikz/ca/cabase});
    \path[ca,at calevels] (0,0) ellipse
    (\pgfkeysvalueof{/tikz/ca/cabase}+\pgfkeysvalueof{/tikz/ca/calevelsx}*\pgfkeysvalueof{/tikz/ca/castep}
    and \pgfkeysvalueof{/tikz/ca/cabase}+\pgfkeysvalueof{/tikz/ca/calevelsy}*\pgfkeysvalueof{/tikz/ca/castep});
  }{
  \end{scope}
}
\newcommand{\CArc}[1]{%
  \global\def\CAinterval{}%
  \begin{scope}[interval/.append code={\global\def\CAinterval{true}},ca,#1]%
    \pgfkeys{/tikz/ca/cabase/.get=\CAbase}%
    \pgfkeys{/tikz/ca/castep/.get=\CAstep}%
    \pgfkeys{/tikz/ca/start/.get=\CAstart}%
    \pgfkeys{/tikz/ca/end/.get=\CAend}%
    \pgfkeys{/tikz/ca/startlevel/.get=\CAstartlevel}%
    \pgfkeys{/tikz/ca/endlevel/.get=\CAendlevel}%
    \pgfkeys{/tikz/ca/label/.get=\CAlabel}%
    \pgfmathparse{\CAbase+\CAstep*\CAstartlevel}%
    \dimdef\CAstartradius{\pgfmathresult pt}%
    \pgfmathparse{\CAbase+\CAstep*\CAendlevel}%
    \dimdef\CAendradius{\pgfmathresult pt}%
    \pgfmathparse{\CAstart*\pgfkeysvalueof{/tikz/ca/caanglestep}}%
    \let\CAstart\pgfmathresult
    \pgfmathparse{\CAend*\pgfkeysvalueof{/tikz/ca/caanglestep}}%
    \let\CAend\pgfmathresult
    \ifdimgreater{\CAend pt}{\CAstart pt}{%
      \edef\tmpa{\pgfkeysvalueof{/tikz/ca/startlabelposauto}}%
      \ifdefstring{\tmpa}{true}{%
        \tikzset{ca,startlabelpos=cw}%
      }{}%
      \edef\tmpa{\pgfkeysvalueof{/tikz/ca/endlabelposauto}}%
      \ifdefstring{\tmpa}{true}{%
        \tikzset{ca,endlabelpos=ccw}%
      }{}%
    }{}%
    \pgfmathparse{\CAstart\pgfkeysvalueof{/tikz/ca/startlabelanchoroffset}}%
    \let\CAstartlabelanchor\pgfmathresult
    \pgfmathparse{\CAend\pgfkeysvalueof{/tikz/ca/endlabelanchoroffset}}%
    \let\CAendlabelanchor\pgfmathresult
    \edef\CAnode{\pgfkeysvalueof{/tikz/ca/labelangle}}%
    \ifdefempty{\CAnode}{%
      \edef\CAnode{\pgfkeysvalueof{/tikz/ca/labelpos}}%
      \pgfmathparse{\CAstart+(\CAend-\CAstart)*\CAnode}%
      \let\CAnode\pgfmathresult
    }{%
      \edef\tmpa{\pgfkeysvalueof{/tikz/ca/caanglestep}}%
      \pgfmathparse{\CAnode*\tmpa}%
      \let\CAnode\pgfmathresult
    }%
    \ifdimequal{\CAstart pt}{\CAend pt}{%
      \dimdef\CAcirclesep{1pt+1.75\pgflinewidth}%
      \fill (\CAstart:\CAstartradius)
      circle (1pt+1.75\pgflinewidth);
      \ifdefempty{\CAlabel}{}{%
        \ifcsstring{tikz@auto@anchor@direction}{left}{%
          \path
          (canvas polar cs:angle=\CAstart+1,radius=\CAstartradius-1pt-1.75\pgflinewidth) --
          (canvas polar cs:angle=\CAstart-1,radius=\CAstartradius-1pt-1.75\pgflinewidth)
          node[midway,auto,ca,every label] {\CAlabel};
        }{%
          \path
          (canvas polar cs:angle=\CAstart+1,radius=\CAstartradius+1pt+1.75\pgflinewidth) --
          (canvas polar cs:angle=\CAstart-1,radius=\CAstartradius+1pt+1.75\pgflinewidth)
          node[midway,auto,ca,every label] {\CAlabel};
        }%
      }%
    }{%
      \dimdef\CAcirclesep{0pt}%
      \ifdimequal{\CAstartradius}{\CAendradius}{%
        \draw[ca,every carc] (\CAstart:\CAstartradius)
        arc[start angle=\CAstart,end angle=\CAend,radius=\CAstartradius];
      }{%
        \draw[ca,every carc,-] plot[smooth,variable=\t,domain=\CAstart:\CAend]
        ({cos(\t)*(\CAstartradius+(\t-\CAstart)/(\CAend-\CAstart)*(\CAendradius-\CAstartradius))},
        {sin(\t)*(\CAstartradius+(\t-\CAstart)/(\CAend-\CAstart)*(\CAendradius-\CAstartradius))});
        \ifdefempty{\CAinterval}{}{%
          \draw (\CAstart:\CAstartradius-1pt-1.25\pgflinewidth) --
          (\CAstart:\CAstartradius+1pt+1.25\pgflinewidth);
          \draw (\CAend:\CAendradius-1pt-1.25\pgflinewidth) --
          (\CAend:\CAendradius+1pt+1.25\pgflinewidth);
        }
      }%
      \pgfmathparse{\CAstartradius+(\CAnode-\CAstart)/(\CAend-\CAstart)*(\CAendradius-\CAstartradius)}%
      \edef\CAnoderadius{\pgfmathresult pt}%
      \ifdefempty{\CAlabel}{}{%
        \path
        (canvas polar cs:angle=\CAnode+1,radius=\CAnoderadius) --
        (canvas polar cs:angle=\CAnode-1,radius=\CAnoderadius)
        node[midway,auto,ca,every label] {\CAlabel};
      }%
      \ifdrawcapoints
        \fill (canvas polar cs:angle=\CAnode,radius=\CAnoderadius) circle (1pt+1.75\pgflinewidth);
      \fi
    }%
    \path (\CAstart:\CAstartradius) node[outer sep/.expanded=\pgfkeysvalueof{/tikz/ca/startlabelsep},anchor/.expanded=\CAstartlabelanchor,ca,every startlabel] {\pgfkeysvalueof{/tikz/ca/startlabel}};
    \path (\CAend:\CAendradius) node[outer sep/.expanded=\pgfkeysvalueof{/tikz/ca/endlabelsep},anchor/.expanded=\CAendlabelanchor,ca,every endlabel] {\pgfkeysvalueof{/tikz/ca/endlabel}};
  \end{scope}%
}

\makeatother
\tikzset{interval/.style={interval-interval,shorten >=-.5\pgflinewidth,shorten <=-.5\pgflinewidth}}

\newcommand{\circl}{\mathbb{C}}

\newcommand{\function}[2]{\colon #1 \rightarrow #2}

\newcommand{\setdef}[2]{\left\{ \hspace{0.5mm} #1: \hspace{0.5mm} #2 \right\}}
\renewcommand{\refeq}[1]{(\ref{eq:#1})}
\newcommand{\Set}[1]{\big\{ #1 \big\}}
\newcommand{\calH}{\ensuremath{\mathcal{H}}}
\newcommand{\calA}{\ensuremath{\mathcal{A}}}
\newcommand{\calB}{\ensuremath{\mathcal{B}}}
\newcommand{\calC}{\ensuremath{\mathcal{C}}}
\newcommand{\calI}{\ensuremath{\mathcal{I}}}
\newcommand{\calJ}{\ensuremath{\mathcal{J}}}
\newcommand{\calK}{\ensuremath{\mathcal{K}}}
\newcommand{\calS}{\ensuremath{\mathcal{S}}}
\newcommand{\calT}{\ensuremath{\mathcal{T}}}
\newcommand{\calO}{\ensuremath{\mathcal{O}}}
\newcommand{\bbI}[1]{\mathbb{I}(#1)}
\newcommand{\cliq}[1]{{\calC(#1)}}
\newcommand{\bund}[1]{{\calB(#1)}}
\newcommand{\tied}{\bowtie}

\newcommand{\complexityclass}[1]{\textsf{\upshape #1}}

\newcommand{\AC}[1]{\ensuremath{\complexityclass{AC}^{\complexityclass{#1}}}}

\makeatletter
\def\@gifnextchar#1#2#3{\let\@tempe#1\def\@tempa{#2}\def\@tempb{#3}%
  \futurelet\@tempc\@gifnch}

\def\@gifnch{\ifx\@tempc\@sptoken\let\@tempd\@tempb%
  \else\ifx\@tempc\@tempe\let\@tempd\@tempa\else\let\@tempd\@tempb\fi\fi\@tempd}

\DeclarePairedDelimiter{\SK@setone}{\lbrace}{\rbrace}
\DeclarePairedDelimiterX{\SK@settwo}[2]{\lbrace}{\rbrace}{#1:#2}
\newcommand{\set}{\@ifstar{\SK@set{*}}{\SK@oset}}
\newcommand{\SK@oset}[1][]{\ifstrempty{#1}{\SK@set{}}{\SK@set{[#1]}}}
\newcommand{\SK@set}[2]{\@gifnextchar\bgroup{\SK@@set{#1}{#2}}{\SK@setone#1{#2}}}
\newcommand{\SK@@set}[3]{\ifstrempty{#1}{%
    \SK@settwo{#2}{#3}%
  }{%
    \SK@settwo#1{#2}{\begin{array}{@{}l@{}}#3\end{array}}%
  }}
\makeatother

\def\afterthmseparator{.}
\makeatletter
\renewcommand{\@begintheorem}[2]{\trivlist
      \item[\hskip \labelsep{\bf #1\ #2\unskip\afterthmseparator}]\em}
\renewcommand{\@opargbegintheorem}[3]{\trivlist
      \item[\hskip \labelsep{\bf #1\ #2\ (#3)\unskip\afterthmseparator}]\em}
\makeatother
\newtheorem{theorem}{Theorem}[section]
\newtheorem{lemma}[theorem]{Lemma}

\newtheorem{definition}[theorem]{Definition}

\newtheorem{remark}[theorem]{Remark}

\newcommand{\bull}{\mbox{$\;\;\;$\vrule height .9ex width .8ex depth -.1ex}}
\newenvironment{proof}{\par\smallbreak\noindent{\bf Proof.~}}
{\unskip\nobreak\hfill \bull \par\medbreak}
\newenvironment{bfenumerate}%
{%
\mbox{}\begin{enumerate}}{\end{enumerate}}

\newcounter{claim}
\renewcommand{\theclaim}{\Alph{claim}}
\newenvironment{claim}{\refstepcounter{claim}%
\par\medskip\par\noindent{\it Claim~\theclaim.~}~\it}%
{\par\smallskip\par}
{$\,\triangleleft$\par\medskip\par}

\newif\ifnotesw\noteswtrue

\usepackage{hyperref}

\title{On the Isomorphism Problem for\\ Helly Circular-Arc Graphs}

\author{Johannes Köbler\quad Sebastian Kuhnert\thanks{Supported by DFG grant  KO 1053/7--1.}\quad 
Oleg Verbitsky\thanks{%
Supported by DFG grant VE 652/1--1.
On leave from the Institute for Applied Problems of Mechanics and Mathematics,
Lviv, Ukraine.}\\[2mm]
\normalsize
Humboldt-Universität zu Berlin,
Institut für Informatik\\ 
\normalsize
Unter den Linden 6,
D-10099 Berlin}

\date{}

\begin{document}

\maketitle

\begin{abstract}
  The isomorphism problem is known to be efficiently solvable for interval
  graphs, while for the larger class of circular-arc graphs its complexity
  status stays open. We consider the intermediate class of intersection graphs
  for families of circular arcs that satisfy the Helly property. We solve the
  isomorphism problem for this class in logarithmic space. If an input graph has
  a Helly circular-arc model, our algorithm constructs it canonically, which
  means that the models constructed for isomorphic graphs are equal.
\end{abstract}

\section{Introduction}

An \emph{intersection representation} of a graph~$G$ is a mapping~$\alpha$ of
the vertex set~$V(G)$ onto a family~$\calA$ of sets such that vertices
$u$~and~$v$ of~$G$ are adjacent if and only if the sets
$\alpha(u)$~and~$\alpha(v)$ have a nonempty intersection. The family~$\calA$ is
called an \emph{intersection model} of~$G$. $G$~is an \emph{interval graph} if
it admits an intersection model consisting of intervals of reals (or,
equivalently, intervals of consecutive integers). The larger class of
\emph{circular-arc (CA) graphs} arises if we consider intersection models
consisting of arcs on a circle. These two archetypal classes of intersection
graphs have important applications, most noticeably in computational genomics,
and have been intensively studied for decades in graph theory and algorithmics;
for an overview see e.g.~\cite{Spinrad}. In general, fixing a class of
admissible intersection models, we obtain the corresponding class of
intersection graphs.

In the \emph{canonical representation problem} for a class~$\mathcal{C}$ of
intersection graphs, we are given a graph $G\in \mathcal{C}$ and have to compute
its intersection representation~$\alpha$ so that isomorphic graphs receive equal
intersection models. This subsumes both recognition of~$\mathcal{C}$ and
isomorphism testing for graphs in~$\mathcal{C}$. In their seminal
work~\cite{BoothL76,LuekerB79}, Booth and Lueker solve both the representation
and the isomorphism problems for interval graphs in linear time. Together with
Laubner, we designed a canonical representation algorithm for interval graphs
that takes logarithmic space~\cite{KoeblerKLV11}. 

The case of CA~graphs remains a challenge up to now. While a circular-arc
intersection model can be constructed in linear time (McConnell~\cite{McC03}),
no polynomial-time isomorphism test for CA~graphs is currently known (though
some approaches~\cite{Hsu95} have appeared in the literature; see the discussion
in~\cite{CurtisLMNSSS13}). A few natural subclasses of CA~graphs have received
special attention among researchers. In particular, for proper CA~graphs both
the recognition and the isomorphism problems are solved in linear time,
respectively, 
in~\cite{LinSS08} and in~\cite{CurtisLMNSSS13}, and in logarithmic space in~\cite{KoeblerKV12}.
The latter result actually gives an logspace algorithm for canonical
representation of proper CA~graphs, and such an algorithm is also known for unit
CA~graphs~\cite{Sou14}. 
The history of the isomorphism
problem for circular-arc graphs is surveyed in more detail by
Uehara~\cite{Ueh13}.

Here we are interested in the class of \emph{Helly circular-arc (HCA) graphs}.
Those are graphs that admit circular-arc models having the \emph{Helly property},
which requires that every family of arcs with nonempty pairwise intersections has a nonempty
overall intersection. Obeying this property is assumed in the representation
problem for HCA~graphs. Since any family of intervals has the Helly property,
and the circles of length at least~$4$ are~HCA but not interval,
the canonical representation problem for HCA~graphs generalizes the canonical
representation problem for interval graphs. On the other hand, not every
CA~model is Helly; see Fig.~\ref{fig:non-hca} for examples. Joeris et~al.\
characterize HCA~graphs among CA~graphs by a family of forbidden induced
subgraphs~\cite{JLMSS11}.

\begin{figure}
  \centering
  \begin{tikzpicture}[baseline=0]
    \node[left,inner sep=0pt] at (-1.25,1.25) {(a)};
    \begin{camodel}[caanglestep=30,rotate=-15,castep=.15cm]
      \CArc{start=0,end=7,startlevel=0,endlevel=7/4}
      \CArc{start=4,end=11,startlevel=0,endlevel=7/4}
      \CArc{start=8,end=15,startlevel=0,endlevel=7/4}
      \CArc{start=1,end=2,startlevel=9/4,endlevel=10/4}
      \CArc{start=5,end=6,startlevel=9/4,endlevel=10/4}
      \CArc{start=9,end=10,startlevel=9/4,endlevel=10/4}
    \end{camodel}
    \begin{scope}[xshift=2.5cm,scale=.9,
      gv/.style={fill,circle,inner sep=1.5pt},
      ge/.style={draw,shorten <=1pt,shorten >=1pt}]
      \node[gv] (1) at (90:1cm) {};
      \node[gv] (2) at (210:1cm) {} edge[ge] (1);
      \node[gv] (3) at (330:1cm) {} edge[ge] (1) edge[ge] (2);
      \node[gv] (12) at (150:1cm) {} edge[ge] (1) edge[ge] (2);
      \node[gv] (23) at (270:1cm) {} edge[ge] (2) edge[ge] (3);
      \node[gv] (31) at (30:1cm) {} edge[ge] (3) edge[ge] (1);
    \end{scope}
  \end{tikzpicture}\hfil
    \begin{tikzpicture}[baseline=0]
    \node[left,inner sep=0pt] at (-1.25,1.25) {(b)};
    \begin{camodel}[caanglestep=30,rotate=15,castep=.15cm]
      \CArc{start=0,end=5,startlevel=0,endlevel=1.25}
      \CArc{start=4,end=9,startlevel=0,endlevel=1.25}
      \CArc{start=8,end=13,startlevel=0,endlevel=1.25}
      \CArc{start=2,end=3,startlevel=1.5,endlevel=1.75}
      \CArc{start=6,end=7,startlevel=1.5,endlevel=1.75}
      \CArc{start=10,end=11,startlevel=1.5,endlevel=1.75}
    \end{camodel}
    \begin{scope}[xshift=2.75cm,scale=.9,
      gv/.style={fill,circle,inner sep=1.5pt},
      ge/.style={draw,shorten <=1pt,shorten >=1pt}]
      \node[gv] (1) at (90:1cm) {};
      \node[gv] (2) at (210:1cm) {} edge[ge] (1);
      \node[gv] (3) at (330:1cm) {} edge[ge] (1) edge[ge] (2);
      \node[gv] (4) at (90:1.5cm) {} edge[ge] (1);
      \node[gv] (5) at (210:1.5cm) {} edge[ge] (2);
      \node[gv] (6) at (330:1.5cm) {} edge[ge] (3);
    \end{scope}
  \end{tikzpicture}
  \caption{Two non-Helly CA~models and their intersection graphs. The graph
    in~(a) admits an HCA~model, while the graph in~(b) does
    not.}\label{fig:non-hca}
\end{figure}
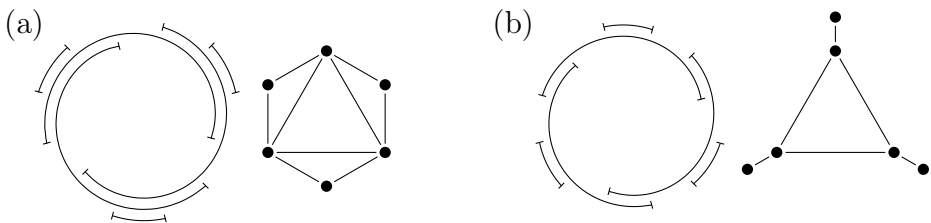

HCA~graphs were introduced by Gavril under the name of \textTheta~circular-arc
graphs~\cite{Gavril74}. Gavril gave an $O(n^3)$~time representation algorithm
for HCA~graphs. Hsu improved this to $O(nm)$~\cite{Hsu95}. Recently, Joeris
et~al.\ gave a linear time representation algorithm~\cite{JLMSS11}. The fastest
known isomorphism algorithm for HCA~graphs is due to Curtis et~al.\ and works in
linear time~\cite{CurtisLMNSSS13}. Chen gave a parallel
$\AC2$~algorithm~\cite{Chen96}.

We aim at designing space efficient algorithms.
In~\cite{KoeblerKV13} we already presented a logspace canonical representation algorithm 
for HCA~graphs. Our approach in~\cite{KoeblerKV13} uses techniques developed by 
McConnell in~\cite{McC03}, and the algorithm is rather intricate. 
Now we suggest an alternative approach that is independent of~\cite{McC03}.
The new algorithm admits a much simpler analysis and exploits some new ideas that may
be of independent interest.

\begin{theorem}\label{thm:main}
  The canonical representation problem for the class of Helly circular-arc graphs
is solvable in logspace.
\end{theorem}

Note that solvability in
logspace implies solvability in logarithmic time by a CRCW~PRAM with
polynomially many parallel processors, i.e., in $\AC1$.
Prior to our work, no $\AC1$ algorithm was known for recognition and 
isomorphism testing of HCA~graphs.

In general, solvability of the isomorphism problem for a non-trivial
class of graphs in logarithmic space is an interesting result because
the general graph isomorphism problem is known to be $\complexityclass{DET}$-hard~\cite{Toran04}
and, therefore, $\complexityclass{NL}$-hard. It is also interesting that for some classes
of intersection graphs, the isomorphism problem is as hard as in general.
For example, Uehara~\cite{Uehara08} shows this for intersection graphs
of axis-parallel rectangles in the plane. Note that any family of
such rectangles has the Helly property.

\paragraph{Our strategy.}
Recall that a hypergraph~$\calH$ is \emph{interval} (resp.\ \emph{circular-arc}) if
it is isomorphic to a hypergraph whose hyperedges are intervals of integers
(resp.\ arcs of a discrete circle). Such an isomorphism is called an interval
(resp.\ arc) representation of~$\calH$.
The overall idea of our algorithm is, like in our approach
to interval graphs in~\cite{KoeblerKLV11}, to exploit the relationship
between an input graph~$G$ and the dual of its maxclique hypergraph,
which will be denoted by~$\bund G$. Fulkerson and Gross~\cite{FulkersonG65} established
that $G$~is an interval graph iff $\bund G$~is an interval hypergraph.
Moreover, represented as an interval system, $\bund G$~can serve as an intersection
model of~$G$. Our approach in~\cite{KoeblerKLV11} consists, therefore, of two
steps: first, construct~$\bund G$ (or, equivalently, find all maxcliques in~$G$)
and, second, design a canonical representation algorithm for interval \emph{hypergraphs}
and apply it to~$\bund G$. The first step is implementable in logspace
because all connected interval graphs are \emph{maximal clique irreducible},
which means that every maxclique~$C$ contains an edge~$uv$ that is contained in
no other maxclique and, therefore, $C$~is equal to the common neighborhood of
$u$~and~$v$.

The Fulkerson-Gross theorem is extended to the class of HCA~graphs by Gavril~\cite{Gavril74}:
$G$~is a HCA~graph iff $\bund G$~is a CA~hypergraph. Also in this case, $\bund
G$~can serve as an isomorphic image of an intersection model for~$G$.
The canonical representation problem for CA~hypergraphs is solved in logspace in~\cite{KoeblerKV12}.
However, the similarity between interval and HCA~graphs ends there because
HCA~graphs are in general not maximal clique irreducible.

Though we are not able to find all maxcliques of an HCA~graph~$G$ directly,
the discussion above shows that the \emph{canonical representation} problem
for HCA~graphs is logspace reducible to the \emph{representation} problem,
where we need just to construct an HCA~representation
and do not need to take care of canonicity. Indeed, once
we have an arbitrary HCA~model of an input graph~$G$, we get all maxcliques of~$G$
by inspection of the sets of arcs sharing a common point.
After all the maxcliques are found, we form the hypergraph~$\bund G$
and compute its canonical representation according to~\cite{KoeblerKV12}
(the details are given in Section~\ref{s:reducing}).

It remains to explain how we compute an HCA~representation~$\alpha$ of~$G$.
It is handy to assume that, if $G$~has $n$~vertices, then its HCA~model~$\alpha(G)$
has $2n$~points, and that no arc in~$\alpha(G)$ shares extreme points with others.
Given $C\subset V(G)$, let~$\alpha^C(G)$ denote the arc system obtained
from~$\alpha(G)$ by flipping the arc~$\alpha(v)$ for all $v\in C$, that is,
by replacing~$\alpha(v)$ with the other arc on the same circle that has
the same extreme points. We make use of a simple consequence of the Helly property:
If $C$~is a maxclique, then $\alpha^C(G)$~becomes an \emph{interval} system.
As was said, we cannot find all maxcliques of~$G$ at once.
However, we are able to find \emph{one} of them, which will be used
for the flipping operation. Our next goal is to compute 
the interval system~$\alpha^C(G)$ up to isomorphism.
Once this is done, we obtain the desired~$\alpha$
(or its isomorphic version) by performing the $C$-flipping for~$\alpha^C(G)$
(note that $(\alpha^C)^C=\alpha$). The flipping operation is
considered in detail in Section~\ref{s:flipping}.

The interval system~$\alpha^C(G)$ is constructed as follows.
In Section~\ref{s:sharpening} we argue that we always can suppose that
$\alpha$~has an additional property: If two arcs intersect and cover the whole
circle, then each of the arcs contains both extreme points of the other.
Under this assumption we are able to compute the \emph{pairwise-intersection matrix}
$M_\alpha=(m_{uv})$, defined by $m_{uv}=|\alpha(u)\cap\alpha(v)|$, and then
also the pairwise-intersection matrix~$M_{\alpha^C}$ for the interval system~$\alpha^C(G)$.
Afterwards we use another result of Fulkerson and Gross saying that an interval system
is determined by its pairwise-intersection matrix up to isomorphism~\cite{FulkersonG65}.
Moreover, it can be reconstructed from the pairwise-intersection matrix
in logspace by an algorithm worked out in~\cite{KoeblerKW12}; see Section~\ref{sec:inters-sizes}.

The pairwise-intersection matrix~$M_\alpha$ is computed in Section~\ref{s:repr}.
The computation is based on the fact that any arc model~$\alpha(G)$ 
is, in a sense, close to~$\bund G$ and on some generic relations
between~$\bund G$ and the closed neighborhood hypergraph
of~$G$, that we explore in Section~\ref{s:bundle}.

\section{Formal definitions}

\paragraph{Hypergraphs.}

Recall that a \emph{hypergraph} is a pair
$(X,\calH)$, where $X=V(\calH)$ is a set of vertices and 
$\calH$~is a family of subsets of~$X$, called \emph{hyperedges}.
We will use the same notation~$\calH$ to denote a hypergraph and its hyperedge set.
A hypergraph has the \emph{Helly property}
if every set of pairwise intersecting hyperedges has a common vertex.
An isomorphism from
a hypergraph~$\calH$ to a hypergraph~$\calK$ is a bijection
$\phi\function{V(\calH)}{V(\calK)}$ such that
$H\in\calH$ iff $\phi(H)\in\calK$ for every $H\subseteq V(\calH)$.

\paragraph{Arc systems.}

For $n\ge3$, consider the directed cycle
on the vertex set $\{1,\dots,n\}$
with arrows from $i$~to~\mbox{$i+1$} and from~$n$ to~$1$.
An \emph{arc}~$A=[a,b]$ consists of the points
appearing in the directed path from~$a$ to~$b$.
The arc $A=\{1,\dots,n\}$ is called \emph{complete}.
If $A=[a,b]$ is not complete, $a$~and~$b$ are referred to as 
\emph{extreme points} of~$A$, the \emph{start point} and the \emph{end point} respectively.
An \emph{arc system}~$\calA$ is a hypergraph 
on the vertex set $\{1,\dots,n\}$ whose hyperedges are arcs.

\paragraph{Arc representations of hypergraphs.}

An \emph{arc representation} of a hypergraph~$\calH$ is an isomorphism~$\rho$
from~$\calH$ to an arc system~$\calA$.
It can be thought of as a circular ordering of~$V(\calH)$ where
every hyperedge is a segment of consecutive vertices.
The arc system~$\calA$ is referred to as an \emph{arc model} of~$\calH$.
The notions of an \emph{interval representation} and an \emph{interval model} of a hypergraph 
are introduced similarly, where \emph{interval} means an interval of integers. 
Hypergraphs having arc representations are called \emph{circular-arc (CA) hypergraphs},
and those having interval representations are called \emph{interval hypergraphs}.

A \emph{representation scheme} is a function defined on CA~hypergraphs
that on input~$\calH$ outputs an arc representation~$\rho_\calH$ of~$\calH$.
A representation scheme is called \emph{canonical} if isomorphic hypergraphs
$\calH\cong\calK$ always receive \emph{equal} arc models $\rho_\calH(\calH)=\rho_\calK(\calK)$.
In~\cite{KoeblerKV12} we designed a canonical representation scheme for CA~hypergraphs
computable in logarithmic space.
Moreover, our algorithm works for 
hypergraphs with multi-hyperedges (the \emph{multiplicity}~$c(H)$
of a hyperedge~$H$ has to be preserved under isomorphisms).

\paragraph{Graphs.}

 The vertex set of a graph~$G$ is denoted by~$V(G)$. 
The \emph{closed neighborhood}~$N[v]$ of a vertex~$v$
consists of~$v$ itself and all vertices adjacent to it. 
A vertex~$u$ is \emph{universal} if $N[u]=V(G)$. 
Two vertices~$u$ and~$v$ are \emph{twins}
if $N[u]=N[v]$. Note that twins are always adjacent.
The \emph{twin class}~$[v]$ of a vertex~$v$ consists of~$v$ itself
along with all its twins. Between two different twin classes there are
either none or all possible edges.
This allows us to consider the \emph{quotient graph}~$G'$
on the vertex set $V(G')=\Set{[v]}\relax{}_{v\in V(G)}$ where two distinct
twin classes $[v]$~and~$[u]$ are adjacent if $v$~and~$u$ are adjacent in~$G$.
The map $v\mapsto[v]$ from~$G$ to~$G'$ 
will be referred to as the \emph{quotient map}.

The \emph{intersection graph}~$\bbI\calH$ of a hypergraph~$\calH$ has the hyperedges
of~$\calH$ as vertices, and two such vertices $A,B\in\calH$ are adjacent if 
$A\cap B\neq\emptyset$. If $\calH$ has hyperedges of multiplicity greater
than~$1$, they become twins in~$\bbI\calH$.

\paragraph{Arc representations of a graph.}

An \emph{intersection representation of a graph}~$G$ is an
isomorphism~$\alpha\function{V(G)}{\calA}$ from~$G$ to the intersection
graph~$\bbI\calA$ of a hypergraph~$\calA$. The hypergraph $\calA$ is then called
an \emph{intersection model} of $G$. If $\calA$ is an arc system, we speak of
\emph{arc representation} and \emph{arc model} of $G$.
Graphs having arc representations are called \emph{circular-arc (CA) graphs}.
In other words, those are graphs isomorphic to the intersection graphs
of CA~hypergraphs. \emph{Helly circular-arc (HCA) graphs} are graphs having
\emph{Helly arc representations}, i.e., representations
providing arc models that obey the Helly property.

It is practical to allow an arc model $\calA$ to have multi-arcs and to require
that an arc representation of a graph $G$ maps twins in $G$ to the same arc in $\calA$
(of multiplicity more than 1). Also, one can require that universal vertices
of $G$ are mapped to the complete arc. Unless stated otherwise, we will consider
arc representations of this kind.
This causes no loss of generality as any such representation can be made injective
in logarithmic space.

A \emph{representation scheme for a class~$\mathcal{C}$ of CA~graphs} is a function 
that on input $G\in \mathcal{C}$ outputs an arc representation~$\alpha_G$ of~$G$.
A representation scheme \emph{for HCA~graphs} must produce Helly arc representations.
If a representation scheme produces equal arc models
for isomorphic input graphs, it is called \emph{canonical}.

\section{The maxclique bundle hypergraph}\label{s:bundle}

An inclusion-maximal clique in a graph~$G$ will be called \emph{maxclique}.
The \emph{maxclique hypergraph}~$\cliq G$ of a graph~$G$ has the same 
vertex set as~$G$ (i.e., $V(\cliq G)=V(G)$) and the maxcliques of~$G$ as its hyperedges.
We now define the \emph{bundle hypergraph}~$\bund G$, which is the dual of~$\cliq G$.
The hypergraph~$\bund G$ has the maxcliques of~$G$
as vertices (i.e., $V(\bund G)=\cliq G$) and a hyperedge~$B_v$ for each 
vertex~$v$ of~$G$, where $B_v$~consists of all maxcliques that contain~$v$. 
We call~$B_v$ the \emph{(maxclique) bundle} of~$v$.

We begin with general properties of the bundle hypergraph that
are true for any graph~$G$.
The first three lemmas summarize well-known facts (see, e.g.,
\cite[Theorem 1.14]{McMc99}); we include short proofs for the reader's convenience.

\begin{lemma}\label{lem:GtoBG}
Define the map $\beta_G\function{V(G)}{\bund G}$ by $\beta_G(v)=B_v$.
Then the correspondence $G\mapsto\beta_G$ is an intersection 
representation scheme for the class of all graphs.
\end{lemma}

\begin{proof}
Note that, for any two distinct vertices $u$~and~$v$,
  \begin{equation}
    \label{eq:BuBv}
 B_u\cap B_v\ne\emptyset\text{ iff $u$ and $v$ are adjacent}.   
  \end{equation}
Indeed, if $C\in B_u\cap B_v$, then both $u$~and~$v$ are in the clique~$C$
and hence adjacent. On the other hand,
if $u$~and~$v$ are adjacent, extend the set $\{u,v\}$ to a maxclique~$C$
and notice that $C\in B_u\cap B_v$. Thus, $\beta_G$~is an intersection
representation of~$G$.
\end{proof}

We now notice that the map~$\beta_G$ in Lemma~\ref{lem:GtoBG}
is, in fact, a Helly representation of the graph~$G$.

\begin{lemma}\label{lem:BGisHelly}
  $\bund G$ is a Helly hypergraph.
\end{lemma}

\begin{proof}
 Suppose that $\Set{B_x}_{x\in X}$ is a family of bundles
with nonempty pairwise intersections. By~\refeq{BuBv}, $X$~is a clique.
Extend $X$ to a maxclique $C$. Then $C\in B_x$ for every $x\in X$.
\end{proof}

Moreover, $\beta_G$ is the smallest possible among all Helly representations
of~$G$ in the following sense. Given a map $\alpha\function{V(G)}{\calH}$
and a set $X\subseteq V(\calH)$, we consider the hypergraph
$\calH|_X=\set*{H\cap X}{H\in\calH}$ on the vertex set~$X$
and define the map $\alpha|_X\function{V(G)}{\calH|_X}$
by $\alpha|_X(v)=\alpha(v)\cap X$.

\begin{lemma}\label{lem:BGisHellyMin}
For every Helly intersection representation $\alpha\function{V(G)}{\calH}$
of a graph~$G$ there is a set $X\subseteq V(\calH)$ such that
$\alpha|_X$~is an intersection representation of~$G$ equivalent
with~$\beta_G$: there is a hypergraph isomorphism~$\psi$ from~$\bund G$
to~$\calH|_X$ such that $\alpha|_X=\psi\circ\beta_G$; see
Fig.~\ref{fig:def-equi}.
\end{lemma}

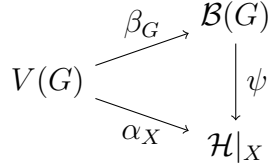
\begin{figure}
  \centering
  \begin{tikzpicture}
    \node (V) at (180:2cm) {$V(G)$};
    \node[minimum width=3em] (B) at (60:1cm)  {$\calB(G)$};
    \node[minimum width=3em] (H) at (-60:1cm) {$\calH|_X$};
    \draw[->] (V) -- node[above] {$\beta_G$} (B);
    \draw[->] (V) -- node[below] {$\alpha_X$} (H);
    \draw[->] (B) -- node[right] {$\psi$} (H);
  \end{tikzpicture}
  \caption{Lemma \protect\ref{lem:BGisHellyMin}:
the intersection representations $\alpha|_X$~and~$\beta_G$ of~$G$
are equivalent up to an isomorphism~$\psi$ between the intersection
models.}\label{fig:def-equi}
\end{figure}

\begin{proof}
For each $C\in\cliq G$, consider the family of hyperedges $\calH_C=\Set{\alpha(v)}_{v\in C}$.
Since $\alpha$~is an intersection representation of~$G$, all pairwise intersections
of the family members are nonempty. Since $\calH$~is a Helly hypergraph,
the overall intersection~$\bigcap\calH_C$ is nonempty. We fix a point $x_C\in\bigcap\calH_C$
and let $X=\set*{x_C}{C\in\cliq G}$. 
Note that $x_C\ne x_{C'}$ if $C\ne C'$
(indeed, the equality $x_C=x_{C'}$ implies that the family $\calH_C\cup\calH_{C'}$
has nonempty overall intersection; therefore, the union $C\cup C'$ of two maxcliques
must be a clique, which is possible only when $C=C'$).
For every $v\in V(G)$, we have
$\alpha(v)\cap X=\set*{x_C}{C\ni v\text{ (or $C\in B_v$)}}$. Hence,
$\psi(C)=x_C$ is an isomorphism from~$\bund G$ to~$\calH|_X$ with
the desired property.
\end{proof}

The following classical result provides a link between HCA~graphs and
CA~hypergraphs; it is exemplified in Fig.~\ref{fig:bundle}.

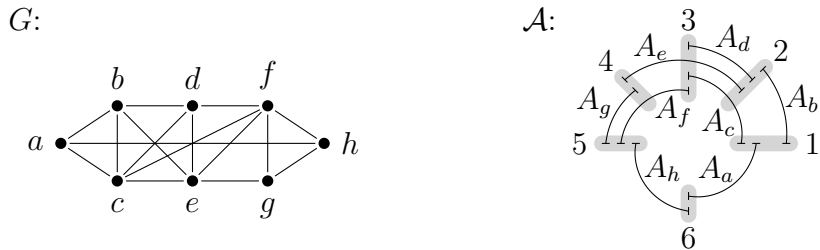
\begin{figure}
  \centering
  \begin{tikzpicture}[baseline=.5cm,
    gv/.style={fill,circle,inner sep=1.5pt},
    ge/.style={draw,shorten <=1pt,shorten >=1pt}]
    \node[left,inner sep=0pt] at (-1,2.15) {$G$:};
    \node[gv,label=left:$a$] (1) at (-.75,.5) {};
    \node[gv,label=above:$b$] (2) at (0,1) {} edge[ge] (1);
    \node[gv,label=below:$c$] (3) at (0,0) {} edge[ge] (1) edge[ge] (2);
    \node[gv,label=above:$d$] (4) at (1,1) {} edge[ge] (2) edge[ge] (3);
    \node[gv,label=below:$e$] (5) at (1,0) {} edge[ge] (2) edge[ge] (3) edge[ge] (4);
    \node[gv,label=above:$f$] (6) at (2,1) {} edge[ge] (3) edge[ge] (4) edge[ge] (5);
    \node[gv,label=below:$g$] (7) at (2,0) {} edge[ge] (5) edge[ge] (6);
    \node[gv,label=right:$h$] (8) at (2.75,.5) {} edge[ge] (6) edge[ge] (7) edge[ge] (1);
  \end{tikzpicture}\hfil
  \begin{tikzpicture}[baseline=0cm]
    \node[left,inner sep=0pt] at (-1.75,1.65) {$\calA$:};
    \begin{scope}[rotate=-45]
      \foreach \p/\x/\i/\o in {1/1/.65cm/1.35cm,2/2/.75cm/1.44cm,3/3/.65cm/1.35cm,
        4/4/.75cm/1.25cm,5/5/.65cm/1.15cm,7/6/.65cm/.95cm} {
        \draw[line width=6pt,line cap=round,color=gray!33!white] (\p*45:\i) --
        (\p*45:\o);
        \node at (\p*45:\o+.3cm) {$\x$};
      }
    \end{scope}
    \begin{camodel}[cabase=.7cm,caanglestep=45,rotate=-45,
      every label/.style={inner sep=0pt}]
      \CArc{start=-1,end=1,startlevel=0,endlevel=1,label=$A_a$,swap,labelpos=.55}
      \CArc{start=1,end=2,startlevel=3,endlevel=3.5,label=$A_b$,labelpos=.4}
      \CArc{start=1,end=3,startlevel=0,endlevel=1,label=$A_c$,swap,labelpos=.4}
      \CArc{start=2,end=3,startlevel=2.5,endlevel=3,label=$A_d$,labelpos=.65}
      \CArc{start=2,end=4,startlevel=1.5,endlevel=2.5,label=$A_e$,labelpos=.65}
      \CArc{start=3,end=5,startlevel=0,endlevel=1,label=$A_f$,swap,labelpos=.4}
      \CArc{start=4,end=5,startlevel=1.5,endlevel=2,label=$A_g$,labelpos=.6}
      \CArc{start=5,end=7,startlevel=0,endlevel=1,label=$A_h$,swap,labelpos=.45}
    \end{camodel}
  \end{tikzpicture}
  \caption{The graph~$G$ contains the maxcliques $C_1=\{a,b,c\}$,
    $C_2=\{b,c,d,e\}$, $C_3=\{c,d,e,f\}$, $C_4=\{e,f,g\}$, $C_5=\{f,g,h\}$, and
    $C_6=\{a,h\}$. Its bundle hypergraph~$\calB_G$ admits the HCA~model~$\calA$
    via the representation~$\rho\colon\calC(G)\to\{1,2,3,4,5,6\}$ that maps each
    maxclique~$C_i$ to the point~$i$, and thus $\rho(B_v)=A_v$ for each $v\in
    V(G)$. The function $\alpha\colon V(G)\to\calA$ that maps each vertex~$v$ to
    the arc~$A_v$ is an HCA~representation of~$G$.}\label{fig:bundle}
\end{figure}

\begin{lemma}[Gavril~\cite{Gavril74}]\label{lem:Gavril}
$G$ is an HCA~graph iff $\bund G$ is a CA~hypergraph.
\end{lemma}

\begin{proof}
  If $\bund G$~is a CA~hypergraph, consider an arc representation~$\rho$ of~$\bund G$.
By Lemmas~\ref{lem:GtoBG} and~\ref{lem:BGisHelly}, $\beta_G$~is a Helly 
intersection representation of~$G$. It remains to notice that $\rho\circ\beta_G$ is 
a Helly arc representation of this graph.

Conversely, assume that $G$~is an HCA~graph and consider a Helly arc representation
$\alpha\function{V(G)}{\calA}$ of~$G$. By Lemma~\ref{lem:BGisHellyMin},
$\bund G$~is isomorphic to the hypergraph~$\calA|_X$ for some set of points~$X$.
For any arc system~$\calA$ and for any set of points 
$X\subseteq V(\calA)$, the hypergraph~$\calA|_X$ is~CA.
Therefore, $\bund G\cong\calA|_X$ is~CA as well.
\end{proof}

The last lemma of the section describes local similarity between the bundle hypergraph~$\bund G$ and 
the \emph{closed neighborhood hypergraph}~$\Set{N[v]}_{v\in V(G)}$.
It shows that the set-theoretic relations between maxclique bundles can be
understood in terms of the adjacency relation of the graph.

\begin{lemma}\label{lem:BvsN}
  Let $u$ and $v$ be arbitrary vertices of a graph~$G$.
  \begin{bfenumerate}
  \item 
$B_u\cap B_v\ne\emptyset$ iff $u\in N[v]$ and iff $v\in N[u]$.
\item 
$B_u\subseteq B_v$ iff $N[u]\subseteq N[v]$.
\item 
Suppose that $u$ and $v$ are adjacent. Then
$B_u\cup B_v=\cliq G$ iff the following three conditions are met:
\begin{enumerate}[ref=(\alph*)]
\item\label{cond:cover}
$N[u]\cup N[v]=V(G)$;
\item 
$w\in N[u]\setminus N[v]$ implies $N[w]\subseteq N[u]$;
\item\label{cond:non-contained}
$w\in N[v]\setminus N[u]$ implies $N[w]\subseteq N[v]$.
\end{enumerate}
\item 
Suppose that $u$ and $v$ are adjacent. Then
$B_u\cap B_v\subseteq B_w$ iff $N[u]\cap N[v]\subseteq N[w]$.
  \end{bfenumerate}
\end{lemma}

\begin{proof}
{\bf 1} readily follows from~\refeq{BuBv}. 

{\bf 2.}
$(\Longrightarrow)$
In this direction, the claim readily follows from Part~1. Indeed, if $x\in N[u]$, 
then $B_x$~intersects~$B_u$ and, hence, also~$B_v$. Therefore, $x\in N[v]$.

$(\Longleftarrow)$
Suppose that $C\in B_u$, that is, $u\in C$. It follows that $C\subseteq N[u]$
and, by assumption, also $C\subseteq N[v]$. This implies that $C\cup\{v\}$
is a clique. Since the clique~$C$ is maximal, $v\in C$, that is, $C\in B_v$.

{\bf 3.}
$(\Longrightarrow)$
Again, this direction follows from Part~1, even without the assumption that
$u$~and~$v$ are adjacent. 

{\bf (a)}
For any~$x$, the bundle~$B_x$ intersects at least one of the bundles $B_u$~and~$B_v$.
Therefore, $x$~belongs to one of the neighborhoods $N[u]$~or~$N[v]$.

{\bf (b)}
Assume that $w\in N[u]\setminus N[v]$. This implies, in particular, that
$B_w$~is disjoint from~$B_v$. If follows from $B_u\cup B_v=\cliq G$
that $B_w\subseteq B_u$. 
By part 2, we conclude that $N[w]\subseteq N[v]$.

{\bf (c)} is symmetric to~(b).

$(\Longleftarrow)$
For this direction, the assumption that $u$~and~$v$ are adjacent is essential.
Assuming that $B_u\cup B_v\ne\cliq G$, we will infer that at least one of
the conditions (a)~and~(b) is false. Indeed, let~$C$ be a maxclique that does
not belong to $B_u\cup B_v$, that is, $u\notin C$ and $v\notin C$.
Since $C$~is inclusion-maximal, it contains a vertex~$x$ non-adjacent to~$v$
and a vertex~$y$ non-adjacent to~$u$. Suppose that (a)~is true. Then
$x$~must be adjacent to~$u$ and, similarly, $y$~must be adjacent to~$v$.
Thus, $vuxy$~is an induced cycle of length~4 in~$G$. Now, (b)~is refuted
by taking $w=x$ because $x\in N[u]\setminus N[v]$ while $y\in N[x]\setminus N[u]$.

{\bf 4.}
$(\Longrightarrow)$
Once again, this direction follows from Part~1.
Indeed, let $x\in N[u]\cap N[v]$. It follows that
$B_x$~intersects both $B_u$~and~$B_v$. 
Since $B_u$~and~$B_v$ intersect (because $u$~and~$v$ are adjacent), 
Lemma~\ref{lem:BGisHelly} implies that $B_x$~intersects even
the intersection $B_u\cap B_v$. Since $B_u\cap B_v\subseteq B_w$,
$B_x$~intersects also~$B_w$ and, therefore, $x\in N[w]$.

$(\Longleftarrow)$
For this direction, the assumption that $u$~and~$v$ are adjacent is not needed
(as then $B_u\cap B_v=\emptyset$). 
Let $C\in B_u\cap B_v$, that is, $u\in C$ and $v\in C$.
It follows that the clique~$C$ is contained in both $N[u]$~and~$N[v]$.
Since $C\subseteq N[u]\cap N[v]\subseteq N[w]$, the set $C\cup\{w\}$ is a clique.
Since $C$ is inclusion-maximal, $w\in C$ and $C\in B_w$ as well.
\end{proof}

\section{Getting canonicity for free}\label{s:reducing}

\begin{lemma}\label{lem:reducing}
  The canonical representation problem for HCA~graphs
is logspace reducible to the (not necessarily canonical)
representation problem for HCA~graphs with no twins and no
universal vertices.
\end{lemma}

\begin{proof}
 We first show that the canonical representation problem for HCA~graphs
reduces in logspace to the problem of computing~$\cliq G$, that is,
to finding all maxcliques in a given HCA~graph~$G$.
Indeed, given~$\cliq G$, we can easily construct the bundle hypergraph~$\bund G$
and the mapping~$\beta_G$. As shown in the proof of Lemma~\ref{lem:Gavril},
we can combine~$\beta_G$ with an arc representation~$\rho_{\bund G}$ of the
CA~hypergraph~$\bund G$ and obtain an arc representation
$\alpha_G=\rho_{\bund G}\circ\beta_G$.
If $\rho_{\bund G}$ is chosen according to the logspace-computable canonical representation
scheme for CA~hypergraphs designed in~\cite{KoeblerKV12}, then $G\mapsto\alpha_G$ will be
a canonical representation scheme for HCA~graphs. Indeed, if $G\cong H$, then
$\bund G\cong\bund H$, which implies that $\alpha_G(G)=\rho_{\bund G}(\bund G)$
is equal to $\alpha_H(H)=\rho_{\bund H}(\bund H)$.

Note now that the problem of computing~$\cliq G$ is equivalent to
its restriction to graphs with no twins and no universal vertices.
Indeed, let~$G'$ be obtained from~$G$ by computing its quotient-graph
with respect to the twin-relation and removing the universal vertex~$[u]$ from it
(if $G$~contains a universal vertex~$u$). Given~$\cliq{G'}$, we easily
obtain~$\cliq G$ by inserting~$[u]$ in each maxclique of~$G'$
and by converting each maxclique $\{[v_1],\ldots,[v_k]\}$ of the quotient-graph
to the maxclique $[v_1]\cup\ldots\cup[v_k]$ of the original graph~$G$.

It remains to show that finding~$\cliq G$ reduces to computing
an arbitrary Helly arc representation~$\alpha$ of~$G$.
Given the arc model~$\alpha(G)$, for each point~$x$ of the circle
we can compute the set $C_x=\set*{v\in V(G)}{x\in\alpha(v)}$.
Obviously, $C_x$~is a clique in~$G$. By Lemma~\ref{lem:BGisHellyMin},
among these cliques there are all maxcliques of~$G$.
Since maximality of a given clique is easy to detect,
this allows us to compute all~$\cliq G$.
\end{proof}

\section{A sharpening of a minimal HCA model}\label{s:sharpening}

If two sets $A$~and~$B$ intersect but neither of them includes the other,
we say that they \emph{overlap} and write $A\between B$.
Suppose now that $A$~and~$B$ are arcs on a circle~$\circl$.
If $A\between B$ and $A\cup B\ne\circl$, we say that $A$~and~$B$ \emph{strictly} overlap
and write $A\between^* B$. If, moreover, $A=[a^-,a^+]$ and $a^+\in B$,
we say that $A$~overlaps~$B$ \emph{on the left} (or that $B$~overlaps~$A$ \emph{on the right})
and write $A\preccurlyeq^* B$ in this case.

A system~$\calA$ of $m$~arcs on the $2m$-point circle will be
called \emph{sharp} if all extreme points of the arcs in~$\calA$
are pairwise distinct; in other words, every point of the circle 
is either start or end point of exactly one arc.
Furthermore, let $A,B\in\calA$, $A=[a^-,a^+]$, and $B=[b^-,b^+]$.
If the extreme points of these arcs occur in the circular order~$a^-b^+b^-a^+$,
we say that $A$~and~$B$ form a \emph{circle cover} and write $A\tied B$.
Note that $A\tied B$ exactly when $A\cup B=\circl$ and $A$~contains both
$b^-$~and~$b^+$ (hence, $B$~contains both $a^-$~and~$a^+$).

\begin{definition}\label{def:sharp}\rm
  Let~$\calA$ be an arc system on a circle~$\circl$
with no multi-arcs and no complete arc. Let~$\calA'$ be
another arc systems on a circle~$\circl'$. A bijection $\sigma\function{\calA}{\calA'}$
is a \emph{sharpening} of~$\calA$ if $\calA'$~is sharp and the following
conditions are met for every $A,B\in\calA$:
\begin{enumerate}
\item\label{cond:int}
$A\cap B=\emptyset$ iff $\sigma(A)\cap\sigma(B)=\emptyset$;
\item\label{cond:sub} 
$A\subseteq B$ iff $\sigma(A)\subseteq\sigma(B)$;
\item\label{cond:ov} 
$A\preccurlyeq^* B$ iff $\sigma(A)\preccurlyeq^*\sigma(B)$;
\item\label{cond:cc} 
Let $A,B\ne\circl$ and $A\cap B\ne\emptyset$. Then
$A\cup B=\circl$ iff $\sigma(A)\tied\sigma(B)$.
\end{enumerate}
\end{definition}
Condition~\ref{cond:cc} means that if $A\cup B=\circl$ and $A$~contains
one extreme point of~$B$, then $\sigma(A)$~must contain both extreme points of~$\sigma(B)$.

\begin{lemma}\label{lem:sharp}
Let $G$ be a HCA graph without twins and universal vertices.
Let $\calA$ be an arc model of $\bund G$.
Then $\calA$ can be sharpened to an arc system~$\calA'$ satisfying the Helly property.
%
%
\end{lemma}

\begin{proof}
Note that every point in $\calA$ is an extreme point of some arc
(otherwise removal of a non-extreme point would not change the intersection
graph of $\calA$, nor violate the Helly property, 
whereas Lemma \ref{lem:BGisHellyMin} implies that $\calA$ is a Helly
intersection model of $G$ with the minimum possible number of points).

First of all, we have to make $\calA$ sharp. To this end,
for each pair of successive points $x$~and~$y$
on~$\circl$ we do the following. Suppose that $y$~is the successor of~$x$.
Suppose that $x$~serves as the end point for the arcs $A_1,\ldots,A_k$
and $y$~serves as the start point for the arcs $B_1,\ldots,B_l$. W.l.o.g., assume that
$A_i\subset A_{i+1}$ and $B_j\subset B_{j+1}$.
The arcs $B_1,\ldots,B_l$ will get new pairwise distinct start points $b^-_l,\ldots,b^-_1$
that will be inserted between~$x$ and~$y$ in this order
(here, it may be helpful to view $\circl$ as a continuous circle).
The arcs $A_1,\ldots,A_k$ will get new pairwise distinct end points $a^+_1,\ldots,a^+_k$
that will be inserted between~$x$ and~$y$ in this order. 
In addition to making $\circl$ sharp, we also want
to ensure Condition~\ref{cond:cc} in Definition~\ref{def:sharp} for each pair
$A_i,B_j$. For this purpose, the sequences $a^+_1,\ldots,a^+_k$
and $b^-_l,\ldots,b^-_1$ will interlace as follows.
Note that, if $A_i$~intersects~$B_j$, then $A_i$~intersects also the longer arc~$B_{j+1}$.
This suggests that we put~$a^+_i$ after~$b^-_{j_i}$, where $j_i$~is the minimum index~$j$
such that $A_i$~intersects~$B_j$. 

We do so for all pairs $x,y$, one by one along $\circl$.
Finally, remove all original points of $\circl$ (none of them is any longer extreme).
The resulting arc model is sharp. We have to check the Helly property and 
Conditions~\ref{cond:int}--\ref{cond:cc} in Definition~\ref{def:sharp}.
Let us analyze the outcome of performing the described transformation for a
particular pair~$x,y$.
\begin{itemize}
\item 
Disjoint arcs remain disjoint.
\item 
Any set of arcs with nonempty overall intersection still has
nonempty overall intersection (because every arc either stays the same
or becomes longer in one direction, if considered on the continuous circle). Therefore,
\begin{itemize}
\item 
intersecting arcs remain intersecting, and
\item 
the Helly property is preserved.
\end{itemize}
\item 
The inclusion and the circle cover relations between any two arcs are preserved.
\item 
If the extreme points of arcs $A=[a^-,a^+]$ and $B=[b^-,b^+]$ appear on the circle
in the order $b^-,a^-,b^+,a^+$, this is so also after the transformation
with the only exception that $a^+=x$ and $b^-=y$. In the last case,
the modified versions of $A$ and $B$ form a circle cover. 
\end{itemize}
It follows that Conditions~\ref{cond:int}, \ref{cond:sub}, and \ref{cond:cc} in Definition~\ref{def:sharp}
are fulfilled for~$\calA'$. 

Verification of Condition \ref{cond:ov} requires some more care.
Suppose that $A=[a^-,a^+]$ and $B=[b^-,b^+]$ are strictly overlapping
arcs in $\calA$ and $A\preccurlyeq^* B$. We know that 
the order $a^-,b^-,a^+,b^+$ of their extreme points will be preserved in 
the modified arc system $\calA'$.
However, we still have to check that the modified arcs $A'$ and $B'$
strictly overlap, that is,  to exclude the possibility that the extreme points
$b^+$ and $a^-$ become neighboring points on the underlying cycle of~$\calA'$.

Since $A\preccurlyeq^*B$, the arc $[b^+,a^-]$ contains an inner point $c$.
Let $\calC=\setdef{C\in\calA}{c\in C}$. Since $\calC$ represents a maxclique
in $G$, it must contain an arc $C=[c^-,c^+]$ that is disjoint with the arc $A$.
It remains to note that the point $c^+\in[b^+,a^-]$ lies strictly
between $b^+$ and $a^-$ also in~$\calA'$ and, therefore, $A'\preccurlyeq^*B'$ indeed.
\end{proof}

\begin{remark}\label{rem:normal-stable}\rm
Our notion of sharpening is related to the concepts of a stable arc system
introduced in~\cite{JLMSS11} and of a normalized arc representation
of a graph introduced in~\cite{Hsu95}. A key property of a stable arc system is that no
additional circle-cover pair can be introduced by moving an extreme point,
unless the intersection graph is also changed by this modification.
In particular, a stable arc system
cannot contain any pair of arcs $A=[a^-,a^+]$, and $B=[b^-,b^+]$
such that $a^+$~and~$b^-$ are consecutive points of the circle
and $b^+\in A$. Due to Condition~\ref{cond:cc} in Definition~\ref{def:sharp},
the latter is true also for any sharpened arc system.

In a \emph{normalized} arc representation $\alpha\colon V(G)\to\calA$, the resulting
arc system~$\calA$ must be stable, and the containment between arcs must reflect
the containments between neighborhoods, i.e., $\alpha(u)\subseteq\alpha(v)$ if
and only if $N[u]\subseteq N[v]$.
By Lemma~\ref{lem:Gavril}, every HCA~graph~$G$ has an arc representation
$\alpha\function{V(G)}{\calA}$ of the form $\alpha=\rho\circ\beta_G$, where
$\rho$~is an arc representation of the bundle hypergraph~$\bund G$.
Let us modify~$\alpha$ to another Helly arc representation $\alpha'=\sigma\circ\alpha$ of~$G$,
where $\sigma\function{\calA}{\calA'}$ is a sharpening of~$\calA$.
Lemmas~\ref{lem:BvsN} and~\ref{lem:sharp} imply that $\alpha'$~is normalized.
\end{remark}

\section{Flipping in a sharp arc system}\label{s:flipping}

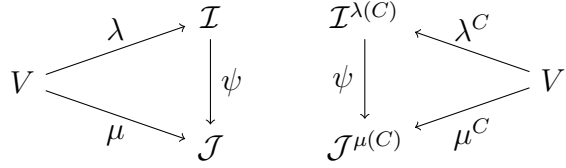
\begin{figure}[t]
  \centering
  \begin{tikzpicture}
    \node (V) at (180:2cm) {$V$};
    \node[minimum width=1.6em] (I) at (60:1cm)  {$\calI$};
    \node[minimum width=1.6em] (J) at (-60:1cm) {$\calJ$};
    \draw[->] (V) -- node[above] {$\lambda$} (I);
    \draw[->] (V) -- node[below] {$\mu$} (J);
    \draw[->] (I) -- node[right] {$\psi$} (J);
  \end{tikzpicture}\qquad
  \begin{tikzpicture}
    \node (V) at (0:2cm) {$V$};
    \node[minimum width=3.2em] (I) at (120:1cm)  {$\calI^{\lambda(C)}$};
    \node[minimum width=3.2em] (J) at (240:1cm) {$\calJ^{\mu(C)}$};
    \draw[->] (V) -- node[above] {$\lambda^C$} (I);
    \draw[->] (V) -- node[below] {$\mu^C$} (J);
    \draw[->] (I) -- node[left] {$\psi$} (J);
  \end{tikzpicture}
  \caption{Lemma \protect\ref{lem:flip}: Flipping preserves isomorphisms that respect extreme points.}\label{fig:flip}
\end{figure}

For the complete arc~$\circl$ the notion of extreme points becomes ambiguous, as then
$\circl=[a,a-1]$ for any $a>1$ and also $\circl=[1,n]$. We will call $[1,n]$ and $[a,a-1]$
a \emph{complete arc with designated extreme points}.
Suppose that an arc $A=[a,b]$ contains more than one point, that is, $a\ne b$.
In this case, we will say that the arc $\tilde A=[b,a]$ is 
obtained from~$A$ by \emph{flipping}. This operation applies, in particular,
to complete arcs with designated extreme points, producing two-point arcs
$[a-1,a]$ or $[n,1]$. If applied to two-point arcs, it produces 
complete arcs with designated extreme points.
Note that flipping preserves sharpness.

Suppose that an arc system~$\calA$ contains no one-point arc
but possibly contains complete arcs with designated extreme points. 
Let $X\subseteq\calA$. 
The \emph{$X$-flipped system} $\calA^X$ is defined
by $\calA^X=\set[\big]{\tilde A}{A\in X}\cup\set[\big]{A}{A\in\calA\setminus X}$.
Given a mapping $\nu\function V\calA$ and a set $C\subseteq V$,
we define the \emph{$C$-flipped mapping} $\nu^C\function V{\calA^{\nu(C)}}$ by
$\nu^C(v)=\widetilde{\nu(v)}$ for $v\in C$ and $\nu^C(v)={\nu(v)}$ for $v\notin C$.

\begin{lemma}\label{lem:flip}
  Let~$\calI$ be an arc system containing no one-point arc but possibly 
complete arcs with designated extreme points.
Let~$\psi$ be a hypergraph isomorphism from~$\calI$ to another arc system~$\calJ$
that takes the extreme points of each arc $A\in\calI$ to the extreme points of the arc $\psi(A)\in\calJ$.
Consider mappings $\lambda\function V\calI$ and $\mu\function V\calJ$
such that $\mu=\psi\circ\lambda$; see Fig.~\ref{fig:flip}.
Let $C\subseteq V$. Then $\psi$~is an isomorphism from~$\calI^{\lambda(C)}$
to~$\calJ^{\mu(C)}$ and $\mu^C=\psi\circ\lambda^C$.
\end{lemma}

\begin{proof}
For every $v\in V$, the isomorphism~$\psi$ maps the arc~$\lambda(v)$ in~$\calI$ onto
the arc~$\mu(v)$ in~$\calJ$. If $\lambda(v)=[a^-,a^+]$ and $\mu(v)=[b^-,b^+]$, then
it is also known that $\psi(\{a^-,a^+\})=\{b^-,b^+\}$. This implies that
$\psi$~maps~$\widetilde{\lambda(v)}$ onto~$\widetilde{\mu(v)}$. Therefore,
$\psi$~maps~$\lambda^C(v)$ onto~$\mu^C(v)$, which means exactly that it is
an isomorphism from~$\calI^{\lambda(C)}$ to~$\calJ^{\mu(C)}$ and $\mu^C=\psi\circ\lambda^C$.
\end{proof}

Lemma~\ref{lem:flip} is true for isomorphisms between arc systems that respect
extreme points. The last condition is not always met. For example,
the transposition~$(23)$, while being an automorphism of the interval system 
$\{[1,3],[2,4]\}$, exchanges extreme points of two different intervals.
However, two isomorphic sharp interval systems always admit an isomorphism
that does respect extreme points. Before we prove this below 
in Lemma~\ref{lem:extreme-respect}, we need to recall some general notions
and facts about interval systems.

A \emph{slot} of a hypergraph~$\calH$ is an inclusion-maximal subset~$S$ of~$V(\calH)$ 
such that each hyperedge of~$\calH$ contains either all of~$S$ or none of it.
Recall that hyperedges $A$~and~$B$ \emph{overlap}, which is denoted as $A\between B$, 
if they intersect but neither of them includes the other.
With respect to the relation~$\between$, any hypergraph~$\calH$
is either connected or is split into \emph{overlap-connected components}.
If $\calO$~and~$\calO'$ are different overlap-connected components, then
either they are vertex-disjoint or
all hyperedges of one of the two components are contained in a single slot
of the other component.\footnote{%
This follows from a simple observation that the conditions $B\subset A$,
$B\between B'$, and $\neg(B'\between A)$ imply that $B'\subset A$.}
If $\calH$~is connected, this containment relation determines a tree-like
decomposition of~$\calH$ into its overlap-connected components.\footnote{%
If $\calH$~is an interval system, this decomposition gives rise to
the concept of a \emph{$PQ$-tree}~\cite{BoothL76}.}
The root in this tree will be referred to as the \emph{top component};
the other components will be called \emph{inner}.
The following fact is due to~\cite[Theorem~2]{ChenY91}; see also~\cite[Section~2.2]{KoeblerKLV11}.

\begin{lemma}[Chen and Yesha~\cite{ChenY91}]\label{lem:yesha}
Suppose that $\calI$~and~$\calJ$ are isomorphic overlap-connected
interval systems.
Let $I_1,\ldots,I_k$ be all slots of~$\calI$ listed in the order as they appear in the line.
Similarly, let $J_1,\ldots,J_k$ be the sequence of slots of~$\calJ$.
Then any isomorphism from~$\calI$ to~$\calJ$ maps either each~$I_s$ onto~$J_s$
or each~$I_s$ onto~$J_{k+1-s}$.  
\end{lemma}

\begin{lemma}\label{lem:extreme-respect}
  Let $\calI$~and~$\calJ$ be isomorphic sharp interval systems. 
For every hypergraph isomorphism~$\psi$ from~$\calI$ to~$\calJ$
there is a hypergraph isomorphism~$\psi'$ from~$\calI$ to~$\calJ$
such that $\psi'(A)=\psi(A)$ for all $A\in\calI$ and, moreover, $\psi'$~respects
extreme points, that is, takes the extreme points of each arc $A\in\calI$ 
to the extreme points of the arc $\psi(A)\in\calJ$.
\end{lemma}

\begin{proof}
We proceed by induction on the number of overlap-connected components of~$\calI$.
In the base case, $\calI$~and~$\calJ$ are overlap-connected.
Using Lemma~\ref{lem:yesha}, we can assume that an isomorphism~$\psi$
from~$\calI$ to~$\calJ$ maps each~$I_s$ onto~$J_s$; the other case 
is symmetric.

We show that, for each $A\in\calI$, the isomorphism~$\psi$ either respects
the extreme points of~$A$ or can be locally modified to respect them.
Let $A=[a^-,a^+]$ and $A=\bigcup_{s=p}^qI_s$. It follows that $\psi(A)=\bigcup_{s=p}^qJ_s$,
$a^-\in I_p$, and $a^+\in I_q$. Moreover, if $\psi(A)=[b^-,b^+]$, then $b^-\in J_p$ and $b^+\in J_q$.

Notice now that, since $\calI$~is sharp, every slot contains at most two points.
Moreover, every two-point slot $[c^-,d^+]$ consists of the start point of some interval~$C$
and the end point of another interval~$D$. The transposition of the points $c^-$~and~$d^+$
violates neither $C$~nor~$D$, nor any other interval.

If $I_p$~is a one-point slot, we immediately conclude that $\psi(a^-)=b^-$.
Suppose that $I_p=[a^-,x^+]$ is a two-point slot. Let $J_p=[b^-,y^+]$.
If $\psi(a^-)=b^-$, we are done.
Otherwise we can ensure $\psi'(a^-)=b^-$ by changing~$\psi$ only on~$I_p$.

In order to ensure that $\psi'(a^+)=b^+$, we may need to modify~$\psi$ on~$I_q$.
In fact, we just need to inspect all two-point slots; if such a slot needs
modification, this will simultaneously fix inconsistency between a pair of start points
and a pair of end points. The analysis of the overlap-connected case is complete.

Suppose now that $\calI$ and $\calJ$ have more than one 
overlap-connected component, that is, are not overlap-connected.
If $\calI$~and~$\calJ$ are disconnected, the claim readily follows
by applying the induction assumption to the corresponding connected
components of $\calI$~and~$\calJ$.

It remains to consider the case when $\calI$~and~$\calJ$ are connected but not
overlap-connected. Assume that an interval $A\in\calI$ contains 
an inner overlap-connected component $\calS\subset\calI$,
then $\psi(V(\calS))\subset\psi(A)$ for any isomorphism~$\psi$ from~$\calI$ to~$\calJ$.
If we remove all points in~$V(\calS)$ from~$\calI$ and all points in~$\psi(V(\calS))$ from~$\calJ$,
the resulting interval systems $\calI'$~and~$\calJ'$ will still contain the extreme points of $A$~and~$\psi(A)$
respectively, and $\psi$~will induce an isomorphism from~$\calI'$ to~$\calJ'$.
By the induction assumption, there are isomorphisms from~$\calI'$ to~$\calJ'$ and from~$\calS$ to~$\psi(\calS)$
that agree with~$\psi$ on hyperedges and respect extreme points.
Merging them, we get the desired isomorphism~$\psi'$ from~$\calI$ to~$\calJ$.
\end{proof}

When we want to apply Lemmas~\ref{lem:flip} and~\ref{lem:extreme-respect},
the interval systems under consideration need to be sharp.
It may happen that we deal with an isomorphic copy of a sharp interval system
that itself is not sharp; consider
for example, $\{[1,4],[2,3]\}$ that is isomorphic to $\{[1,4],[1,2]\}$. 
In such cases the following fact will be helpful.

\begin{lemma}\label{lem:iso-sharpen}
  Suppose that for an interval system~$\calJ$ there is an isomorphic sharp
interval system~$\calJ'$. Then such~$\calJ'$ can be computed in logspace
along with an isomorphism from~$\calJ$ to~$\calJ'$.
\end{lemma}

\begin{proof}
Suppose that~$\calJ$ is isomorphic to a sharp interval system~$\calS$
and $\varphi$~is an isomorphism from~$\calJ$ to~$\calS$. Since~$\calS$
cannot contain any 1-point interval, the same holds true for any
isomorphic system, in particular, for~$\calJ$. Furthermore, $\calJ$~cannot
contain any point that serves simultaneously as the start point of
an interval~$A$ and the end point of another interval~$B$;
otherwise the intervals $\varphi(A)$~and~$\varphi(B)$ in~$\calS$ would also intersect at
only one point and thus share an extreme point.

Given~$\calJ$, we construct an interval system~$\calJ'$ in three steps,
each doable in logspace.
\begin{enumerate}
\item 
Remove all \emph{interior points} from~$\calJ$, that is, those points that are not extreme
for any interval.
\item 
For each point~$x$ that is the start point of two or more intervals
$A_1,\ldots,A_k$, do the following. W.l.o.g., assume that
$A_1\supset A_2\supset\ldots\supset A_k$. Let $y\in A_1$ be the point
next to~$x$. We provide the arcs $A_1,\ldots,A_k$ with new pairwise distinct start 
points $x=a^-_1,\ldots,a^-_k$ that will be inserted between $x$~and~$y$ in this order.
\item 
Do similarly with the shared end points.
\end{enumerate}
Being removed in the first step, interior points never appear later.
The 2nd and the 3rd steps ensure that no two intervals in~$\calJ'$
share an extreme point. Thus, $\calJ'$~is sharp.
The main efforts are needed to show that $\calJ'$~is isomorphic to~$\calJ$.

To prove this, we use induction on the number of overlap-connected components of~$\calJ$.
In the base case, $\calJ$ is overlap-connected. 
Note that Lemma~\ref{lem:yesha} has the following interpretation.

\begin{claim}\label{cl:yesha}
  Let $\calI$~and~$\calI'$ be interval systems isomorphic as hypergraphs.
If they are overlap-connected, then either $\calI=\calI'$
or $\calI'$~is obtained from~$\calI$ by a reflection of the line.
\end{claim}

Claim~\ref{cl:yesha} implies that an overlap-connected~$\calJ$ is
geometrically congruent to~$\calS$ and is, therefore, sharp.
Thus, the algorithm just returns $\calJ'=\calJ$ in this case.

Suppose now that $\calJ$~is disconnected. Note that we can
obtain~$\calJ'$ by applying the algorithm to each connected component
of~$\calJ$ and merging the results. The isomorphism $\calJ'\cong\calJ$
readily follows by the induction assumption.

It remains to consider the case when $\calJ$~is connected but not
overlap-connected. Let~$\calT$ denote the top overlap-connected
component of~$\calJ$. By Claim~\ref{cl:yesha}, $\calT$~is congruent to
the top overlap-connected component of~$\calS$. In particular, no extreme
point is shared by two intervals in~$\calT$ (but $\calT$~has interior points).
Assume first that no extreme point of an interval in~$\calT$ is shared with any interval
in $\calJ\setminus\calT$. 
Then the output~$\calJ'$ is obtainable by leaving~$\calT$ as it is and by applying the algorithm
to the children-components within each slot of~$\calT$.
Note that $\varphi$~maps every inner overlap-connected
component of~$\calJ$ to an inner component of~$\calS$, which is sharp.
Using the induction assumption for each child-component within
each slot of~$\calT$, we conclude that $\calJ'\cong\calJ$ also in this case.

Assume now that there is an extreme point~$x$ of an interval~$\calT$
that is also an extreme point of some interval $A\in\calJ\setminus\calT$. 
Fix~$A$ to be the longest of such intervals. Denote the overlap-connected
component of~$\calJ$ containing~$A$ by~$\calA$. Let~$T$ denote the slot
of~$\calT$ containing~$x$. Note that $\calA$~is one of the children-components
located in~$T$. Looking at the image~$\varphi(T)$ in~$\calS$, we see that $T$~must
contain a point~$z$ not included in any inner overlap-connected component
(namely $z=\varphi^{-1}(z')$ for $z'\in\varphi(T)$ being an extreme point
of an interval in the top component of~$\calS$).
Moving~$z$ to any other place in~$T$ outside the children-components
results in an interval system isomorphic to~$\calJ$. In particular,
we can make~$z$ a new extreme point of an interval in~$\calT$ instead of~$x$.
Denote the resulting interval system by~$\tilde\calJ$. We can, therefore,
obtain the same outcome~$\calJ'$ as follows.
\begin{itemize}
\item 
Remove~$V(\calA)$ from~$\tilde\calJ$ and denote the result by~$\calK$.
Looking at~$\varphi$ on~$V(\calA)$ and on $V(\calJ)\setminus V(\calA)$,
we see that both $\calA$~and~$\calK$ are isomorphic to sharp interval systems.
\item 
Apply the algorithm to $\calA$~and~$\calK$ and denote the outputs by
$\calA'$~and~$\calK'$, respectively.
\item 
Reinsert~$\calA'$ in~$\calK'$ within the corresponding slot.
\end{itemize}
Since $\calA'\cong\calA$ and $\calK'\cong\calK$ by the induction assumption,
we conclude that $\calJ'\cong\calJ$ as claimed.

In general, the algorithm is run on an arbitrary~$\calJ$. After
computing~$\calJ'$ we invoke the algorithm of~\cite{KoeblerKLV11} to find a
hypergraph isomorphism from~$\calJ$ to~$\calJ'$. In the case of failure,
we conclude that the input system~$\calJ$ is not isomorphic to any sharp interval system.
\end{proof}

\section{Pairwise intersections as a complete isomorphism invariant for interval hypergraphs}
\label{sec:inters-sizes}

Given a hypergraph~$\calH$ and a bijection $\nu\function V\calH$, we define
the \emph{pairwise-intersection matrix} $M_\nu=(m_{uv})_{u,v\in V}$ by
$m_{uv}=|\nu(u)\cap\nu(v)|$. If $\psi$~is an isomorphism from~$\calH$ to~$\calK$
and the bijection $\mu\function V\calK$ is defined by $\mu=\psi\circ\lambda$,
then obviously $M_\lambda=M_\mu$. It turns out that the converse is also true
if $\calH$~is an interval hypergraph.

\begin{figure}
  \centering
  \begin{tikzpicture}
    \node (V) at (180:2cm) {$V$};
    \node[minimum width=1.6em] (I) at (60:1cm)  {$\calI$};
    \node[minimum width=1.6em] (J) at (-60:1cm) {$\calJ$};
    \draw[->] (V) -- node[above] {$\lambda$} (I);
    \draw[->] (V) -- node[below] {$\mu$} (J);
    \draw[->] (I) -- node[right] {$\psi$} (J);
  \end{tikzpicture}
  \caption{Lemma \protect\ref{lem:inters-sizes}: If $M_\lambda=M_\mu$ and
    $\calI$~is an interval hypergraph, then
    $\calI\cong\calJ$.}\label{fig:inters-sizes}
\end{figure}
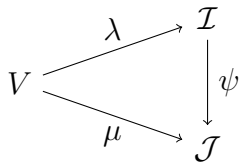

\begin{lemma}[Fulkerson and Gross~\cite{FulkersonG65}]\label{lem:inters-sizes}
  Let~$\calI$ be an interval system and $\calJ$~be an arbitrary hypergraph.
Suppose that $M_\lambda=M_\mu$ for bijections $\lambda\function V\calI$ and
$\mu\function V\calJ$. Then there is a hypergraph isomorphism~$\psi$ such that
$\mu=\psi\circ\lambda$; see Fig.~\ref{fig:inters-sizes}.
\end{lemma}


We will use the fact that $\calI$~and~$\lambda$ are efficiently reconstructible
from a given $M=M_\lambda$.

\begin{lemma}[Köbler, Kuhnert, and Watanabe~\cite{KoeblerKW12}]\label{lem:L-reconst}
  There is a logspace algorithm that, given an integer matrix $M=(m_{uv})_{u,v\in V}$,
constructs an interval system~$\calI$ and a bijection $\lambda\function V\calI$
such that $M=M_\lambda$ or detects that such an interval system does not exist.
\end{lemma}

\section{A representation scheme for HCA~graphs in logspace}\label{s:repr}

We are now prepared to prove Theorem~\ref{thm:main}.
By Lemma~\ref{lem:reducing}, it suffices to design a 
(not necessarily canonical) representation scheme
for HCA~graphs that have no twins and no universal vertices and to show
that this scheme is computable in logspace.

Let~$G$ be an input graph on $n$~vertices. 
We assume that $G$~is~HCA and has neither twins nor universal vertices.
Note that then its bundle hypergraph~$\bund G$
has no multi-hyperedges $B_u=B_v$ and no complete hyperedge $B_u=\cliq G$.
Let $\beta_G\function{V(G)}{\bund G}$ be the Helly intersection representation of~$G$
as defined in Lemma~\ref{lem:GtoBG}.
By Lemma~\ref{lem:Gavril}, $\bund G$~is a CA~hypergraph.
Consider its arbitrary arc representation $\rho\function{\bund G}\calB$.
As it will be beneficial to deal with sharp arc models,
consider an arbitrary sharpening $\sigma\function\calB\calA$ of $\calB$
to a sharp Helly arc system~$\calA$, which exists because $\calB$~contains no
multi-arcs and no complete arc. Define 
\begin{equation}
  \label{eq:alpha}
\alpha=\sigma\circ\rho\circ\beta_G. 
\end{equation}
 Thus, $\alpha\function{V(G)}\calA$
is a Helly arc representation of~$G$ by a sharp arc model~$\calA$.

\begin{lemma}\label{lem:a}
For~$\alpha$ defined by~\refeq{alpha},
the pairwise-intersection matrix~$M_\alpha$ depends on~$G$ only
(and neither on~$\rho$ nor on~$\sigma$) and can be computed in logspace.
\end{lemma}

\begin{proof}
  Consider first $m_{vv}=|\alpha(v)|$. The arc~$\alpha(v)$ contains two its own
extreme points and, furthermore, every vertex~$u$ adjacent to~$v$
contributes one or two extreme points of~$\alpha(u)$ into~$\alpha(v)$.
More precisely, the following configurations are possible.
\begin{description}
\item[$\alpha(u)\subset\alpha(v)$ --- 2 contributed points:] 
By the definition of sharpening, this happens exactly when $B_u\subset B_v$,
which is equivalent to the logspace-verifiable
condition $N[u]\subset N[v]$ by Lemma~\ref{lem:BvsN}.2.
\item[$\alpha(u)\tied\alpha(v)$ --- 2 contributed points:]
By the definition of sharpening, this happens exactly when $B_u\cup B_v=\cliq G$,
which is equivalent to the logspace-verifiable
conditions~\ref{cond:cover}--\ref{cond:non-contained} in Lemma~\ref{lem:BvsN}.3.
\item[$\alpha(u)\between^*\alpha(v)$ --- 1 contributed point:]
the remaining case.
\end{description}

Consider now $m_{uv}=|\alpha(u)\cap\alpha(v)|$ for $u\ne v$.
In the simplest case of non-adjacent $u$~and~$v$ we have $m_{uv}=0$.
Also, $m_{uv}=m_{uu}$ if $\alpha(u)\subset\alpha(v)$ or, equivalently, $N[u]\subset N[v]$.
Similarly, $m_{uv}=m_{vv}$ if $N[v]\subset N[u]$.
Furthermore, $m_{uv}=m_{uu}+m_{vv}-2n$ if $\alpha(u)\tied\alpha(v)$, which is verifiable 
by Lemma~\ref{lem:BvsN}.3.

It remains to compute~$m_{uv}$ if $\alpha(u)\between^*\alpha(v)$. The intersection contains
one extreme point of~$\alpha(u)$ and one of~$\alpha(v)$. Any other vertex~$w$
contributes 0, 1, or 2 extreme points of~$\alpha(w)$. The contribution is~0 when
$\alpha(w)$~is disjoint from $\alpha(u)$~or~$\alpha(v)$ or when it contains
at least one of these arcs. Let us analyze the remaining cases (some cases symmetric
up to swapping $u$~and~$v$ are omitted). The first four conditions are 
verifiable in logspace similarly to the above by Lemma~\ref{lem:BvsN}.
\begin{description}
\item[{\rm $\alpha(w)\subset\alpha(u)$ and $\alpha(w)\subset\alpha(v)$} --- 2 contributed points,]
\item[{\rm $\alpha(w)\subset\alpha(u)$ and $\alpha(w)\between^*\alpha(v)$} --- 1 contributed point,]
\item[{\rm $\alpha(w)\tied\alpha(u)$ and $\alpha(w)\tied\alpha(v)$} --- 2 contributed points,]
\item[{\rm $\alpha(w)\tied\alpha(u)$ and $\alpha(w)\between^*\alpha(v)$} --- 1 contributed point,]
\item[{\rm $\alpha(w)\between^*\alpha(u)$ and $\alpha(w)\between^*\alpha(v)$}:]
This case is more complicated. W.l.o.g., suppose that $\alpha(u)\preccurlyeq^*\alpha(v)$
and, hence, $\rho(B_u)\preccurlyeq^*\rho(B_v)$. 
Note first that the arc configuration $\alpha(v)\preccurlyeq^*\alpha(w)\preccurlyeq^*\alpha(u)$
is non-Helly and, hence, cannot occur. There remain two subcases.
\begin{description}
\item[$\alpha(u)\preccurlyeq^*\alpha(w)\preccurlyeq^*\alpha(v)$ --- 0 contributed points:] 
By the definition of sharpening, this happens exactly when
$\rho(B_u)\preccurlyeq^*\rho(B_w)\preccurlyeq^*\rho(B_v)$, which is equivalent
$\rho(B_u)\cap\rho(B_v)\subset\rho(B_w)$. Since $\rho$~is a hypergraph isomorphism,
the last condition reads $B_u\cap B_v\subset B_w$,
which is equivalent to the logspace-verifiable
condition $N[u]\cap N[v]\subseteq N[w]$ by Lemma~\ref{lem:BvsN}.4.
\item[{\rm $\alpha(w)\preccurlyeq^*\alpha(u)$ and $\alpha(w)\preccurlyeq^*\alpha(v)$ or
$\alpha(u)\preccurlyeq^*\alpha(w)$ and $\alpha(v)\preccurlyeq^*\alpha(w)$}] {\bf--- 1 contributed point:}
This is the complementary subcase.
\end{description}
\end{description}
The analysis is complete. The matrix entry~$m_{uv}$ is obtained by summing up the
contributions of~$\alpha(w)$ over all~$w$.
\end{proof}

Next, we need to find an arbitrary maxclique $C\in\cliq G$.
We have to argue that this is doable in logspace.
An edge~$uv$ in a graph~$G$ is called \emph{essential} if
it is contained in a unique maxclique~$C$.
The following lemma implies that, for each~$uv$, we can check in logspace
if it is essential. If so, the corresponding maxclique~$C$ can be computed
also in logspace as $C=N[u]\cap N[v]$.

\begin{lemma}\label{lem:essential}
  An edge~$uv$ is essential if and only if the intersection $N=N[u]\cap N[v]$
is a clique.
\end{lemma}

\begin{proof}
Note first that any clique containing~$uv$ is included in~$N$.
If $N$~is a clique, this implies that $N$~is actually a maxclique
and, moreover, it is the only maxclique containing~$uv$.

Suppose now that $N$~contains non-adjacent vertices $x$~and~$y$.
Then two triangles $\{u,v,x\}$ and $\{u,v,y\}$ can be extended to
two different maxcliques both containing~$uv$. 
\end{proof}

It is known~\cite{OpsutR81} that if $G$~is a connected interval graph, then
every maxclique in~$G$ contains an essential edge.
This allows to compute the bundle hypergraph~$\bund G$ in logspace, which was an important
ingredient of our canonical representation scheme for interval graphs in~\cite{KoeblerKLV11}.
However, connected HCA~graphs do not enjoy this property;
the Haj\'os (or 3-sun) graph depicted in Fig.~\ref{fig:non-hca}(a) is a counterexample.
Fortunately, every nonempty HCA~graph has at least one maxclique
that can be efficiently found due to the fact that it contains an essential edge.

\begin{lemma}\label{lem:essential-HCA}
  Every nonempty HCA~graph~$G$ contains an essential edge~$uv$.
\end{lemma}

\begin{proof}
  It is enough to prove the lemma for~$G$ with no twins and no universal vertices.
Consider the Helly arc representation $\beta=\rho\circ\beta_G$ of~$G$ where
$\rho$~is an arc representation of the CA~hypergraph~$\bund G$.
Fix~$v$ to be a non-isolated vertex whose maxclique bundle~$B_v$ is minimal under inclusion.
Note that $B_v\cup B_w=\cliq G$ for no vertex~$w\in N[v]$ for else $w$~would be
universal. Thus, for every~$w$ either $B_v\subseteq B_w$ or $B_v\between^* B_w$.
If all~$w\in N[v]$ satisfy the former condition, $N[v]$~is a clique and we are done
(we can choose~$u$ arbitrarily from~$N[v]$). Otherwise fix~$u\in N[v]$ to be
a vertex with $|B_v\cap B_u|$ as small as possible.
Note that $B_v\between^* B_u$. 

It remains to argue that $uv$~is an essential edge.
By Lemma~\ref{lem:essential}, we have to show that
the intersection $N=N[u]\cap N[v]$ is a clique.
Assume, to the contrary, that $N$~contains non-adjacent $x$~and~$y$.
Looking at the arc representation~$\beta$, we see that the arcs $\beta(x)$~and~$\beta(y)$
must intersect the arc $\beta(v)\cap\beta(u)$ from different sides.
One of $\beta(x)$~and~$\beta(y)$ must, therefore, contain the extreme point of~$\beta(v)$
contained in~$\beta(u)$. Without loss of generality, suppose that this is~$\beta(x)$.
It follows that $|\beta(v)\cap\beta(x)|<|\beta(v)\cap\beta(u)|$,
giving a contradiction with the assumption that $|B_v\cap B_u|$ is the smallest possible.
\end{proof}

\begin{lemma}\label{lem:b}
Let~$\alpha$ be a sharp Helly representation of a graph~$G$ without universal vertices.
Let $C\in\cliq G$ be a maxclique in~$G$. Consider the $C$-flipped mapping 
$\alpha^C\function{V(G)}{\calA^{\alpha(C)}}$. Then
  $\calI=\calA^{\alpha(C)}$ is an interval system, that is, there are two consecutive points
$x$~and~$y$ on the circle such that no interval $I\in\calI$ contains both $x$~and~$y$
unless $I=[x,y]$ is the complete arc with designated extreme points $x$~and~$y$
(obtained by flipping the arc $\{x,y\}$).
\end{lemma}

\begin{proof}
  Since $\alpha$~is a Helly representation of~$G$, the arcs in the
  set~$\alpha(C)$ have a common point~$x$.
  Suppose that $x$~is an extreme point of an arc $A\in\alpha(C)$.
  Choosing~$y$ to be the point of~$A$ next to~$x$, we obtain the claimed pair~$x,y$.
\end{proof}

We remark that the sharpness condition in Lemma~\ref{lem:b} is crucial. Indeed,
consider the graph~$G$ and its HCA~representation~$\alpha$ given in
Fig.~\ref{fig:bundle}. The $C_6$-flipped mapping~$\alpha^{C_6}$ results in a
non-interval arc system.

\begin{lemma}\label{lem:c}
Let $\alpha$ be defined by~\refeq{alpha} and
$\lambda=\alpha^C$ for $C\in\cliq G$. 
Then $M_\lambda$~can be computed in logspace from~$M_\alpha$ and~$C$.
\end{lemma}

\begin{proof}
Let $M_\lambda=(m^\lambda_{uv})$ and $M_\alpha=(m^\alpha_{uv})$. We have
$m^\lambda_{vv}=m^\alpha_{vv}$ if $v\notin C$ and $m^\lambda_{vv}=2n+2-m^\alpha_{vv}$ if $v\in C$.
For different $u$~and~$v$, $m^\lambda_{uv}$ is computed by inspection of several cases. 
If $u\notin C$ and $v\notin C$, then $m^\lambda_{uv}=m^\alpha_{uv}$.
If $u\in C$ and $v\notin C$, then\\
\begin{tabular}{lcl}
$\alpha(u)\cap\alpha(v)=\emptyset$ & $\Rightarrow$ & $m^\lambda_{uv}=m^\alpha_{vv}$;\\
$\alpha(u)\subset\alpha(v)$ & $\Rightarrow$  & $m^\lambda_{uv}=m^\alpha_{vv}-m^\alpha_{uu}+2$;\\
$\alpha(u)\supset\alpha(v)$ & $\Rightarrow$  & $m^\lambda_{uv}=0$;\\
$\alpha(u)\tied\alpha(v)$ & $\Rightarrow$ & $m^\lambda_{uv}=2n+2-m^\alpha_{uu}$;\\
$\alpha(u)\between^*\alpha(v)$ & $\Rightarrow$ & $m^\lambda_{uv}=m^\alpha_{vv}-m^\alpha_{uv}+1$.
\end{tabular}\\
The case of $u\notin C$ and $v\in C$ is symmetric.
If $u\in C$ and $v\in C$, then\\
\begin{tabular}{lcl}
$\alpha(u)\cap\alpha(v)=\emptyset$ & $\Rightarrow$ & $m^\lambda_{uv}=2n+4-m^\alpha_{uu}-m^\alpha_{vv}$;\\
$\alpha(u)\subset\alpha(v)$ &  $\Rightarrow$ & $m^\lambda_{uv}=2n+2-m^\alpha_{vv}$;\\
$\alpha(u)\supset\alpha(v)$  & $\Rightarrow$ & $m^\lambda_{uv}=2n+2-m^\alpha_{uu}$;\\
$\alpha(u)\tied\alpha(v)$ & $\Rightarrow$ & $m^\lambda_{uv}=0$;\\
$\alpha(u)\between^*\alpha(v)$  & $\Rightarrow$ & $m^\lambda_{uv}=2n+2+m^\alpha_{uv}-m^\alpha_{uu}-m^\alpha_{vv}$.
\end{tabular}\\
Recall that the relationship between $\alpha(u)$~and~$\alpha(v)$ 
is recognizable by Lemma~\ref{lem:BvsN} and Definition~\ref{def:sharp}.
\end{proof}

Now we can complete the description of our algorithm for computing an Helly arc
representation of the input graph $G$.  Suppose that
$\alpha\function{V(G)}{\calA}$ is a normalized Helly arc representation of $G$. 
What follows does not depend on a particular choice of~$\alpha$.

\begin{description}
\item[\it Step 1.] 
Compute the intersection matrix~$M_\alpha$.
By Lemma~\ref{lem:a}, this matrix can be computed in logspace and does not depend on~$\alpha$.
\item[\it Step 2.] 
Compute a maxclique~$C$ of~$G$. This is doable in logspace
according to Lemmas~\ref{lem:essential} and~\ref{lem:essential-HCA}.
\item[\it Step 3.] 
Compute the intersection matrix $M_\lambda$ for the $C$-flipped mapping $\lambda=\alpha^C$.
This can be done in logspace due to Lemma~\ref{lem:c}.

Note that, by Lemma~\ref{lem:b}, the flipped arc system
$\calI=\calA^{\alpha(C)}$ is actually an interval system.
\item[\it Step 4.] 
Compute an interval system~$\calJ$ and a
mapping $\mu\function{V(G)}\calJ$ such that $M_\mu=M_\lambda$.
For that purpose, we invoke the algorithm of Lemma~\ref{lem:L-reconst}.

Note that, by Lemma~\ref{lem:inters-sizes}, $\calJ$~and~$\calI$ are isomorphic
hypergraphs. Recall that $\lambda$ is a mapping from $V(G)$ to $\calI$.
Lemma~\ref{lem:inters-sizes}, moreover, ensures that there is a hypergraph isomorphism~$\psi$
from~$\calI$ to~$\calJ$ such that
\[
\mu=\psi\circ\lambda.
\]
\item[\it Step 5.] 
Modify $\mu$ and $\calJ$ so that $\calJ$ becomes sharp if it
is not such from the very beginning. 
This is possible due to Lemma~\ref{lem:iso-sharpen} because
$\calI\cong\calJ$ is a sharp interval system.

By Lemma~\ref{lem:extreme-respect} we can assume that $\psi$~respects extreme points
of intervals in $\calI$ and~$\calJ$.
\item[\it Step 6.] 
Now, we ``close'' the interval $1,\ldots,2n$ to the cycle where $1$~succeeds~$2n$
and regard $\calJ$~and~$\calI$ as arc systems, that possibly have complete arcs
with designated extreme points. The mapping~$\psi$ stays a hypergraph isomorphism
respecting extreme points of all arcs.
\item[\it Step 7.] 
Compute the $C$-flipped mapping $\mu^C\function{V(G)}{\calJ^{\mu(C)}}$.
By Lemma~\ref{lem:flip},
\[
\mu^C=\psi\circ\lambda^C=\psi\circ\alpha
\]
and $\psi$~is a hypergraph isomorphism from $\calI^{\lambda(C)}=\calA$ to~$\calJ^{\mu(C)}$.
It follows that, like~$\alpha$, the constructed mapping $\mu^C$~is a Helly arc representation of~$G$.
\end{description}

The proof of Theorem~\ref{thm:main} is complete.


\begin{thebibliography}{10}

\bibitem{BoothL76}
K.~Booth and G.~Lueker.
\newblock Testing for the consecutive ones property, interval graphs, and graph
  planarity using {$PQ$}-tree algorithms.
\newblock {\em J. Comput. Syst. Sci.}, 13(3):335--379, 1976.

\bibitem{Chen96}
L.~Chen.
\newblock Graph isomorphism and identification matrices: Parallel algorithms.
\newblock {\em IEEE Trans. Parallel Distrib. Syst.}, 7(3):308--319, 1996.

\bibitem{ChenY91}
L.~Chen and Y.~Yesha.
\newblock Parallel recognition of the consecutive ones property with
  applications.
\newblock {\em J. Algorithms}, 12(3):375--392, 1991.

\bibitem{CurtisLMNSSS13}
A.~R. Curtis, M.~C. Lin, R.~M. McConnell, Y.~Nussbaum, F.~J. Soulignac, J.~P.
  Spinrad, and J.~L. Szwarcfiter.
\newblock Isomorphism of graph classes related to the circular-ones property.
\newblock {\em Discrete Mathematics \& Theoretical Computer Science},
  15(1):157--182, 2013.

\bibitem{FulkersonG65}
D.~Fulkerson and O.~Gross.
\newblock {Incidence matrices and interval graphs.}
\newblock {\em Pac. J. Math.}, 15:835--855, 1965.

\bibitem{Gavril74}
F.~Gavril.
\newblock Algorithms on circular-arc graphs.
\newblock {\em Networks}, 4(4):357--369, 1974.

\bibitem{Hsu95}
W.-L. Hsu.
\newblock {$O(MN)$} algorithms for the recognition and isomorphism problems on
  circular-arc graphs.
\newblock {\em SIAM J. Comput.}, 24(3):411--439, 1995.

\bibitem{JLMSS11}
B.~L. Joeris, M.~C. Lin, R.~M. McConnell, J.~P. Spinrad, and J.~L. Szwarcfiter.
\newblock Linear time recognition of Helly circular-arc models and graphs.
\newblock {\em Algorithmica}, 59(2):215--239, 2 2011.

\bibitem{KoeblerKLV11}
J.~K{\"o}bler, S.~Kuhnert, B.~Laubner, and O.~Verbitsky.
\newblock Interval graphs: {C}anonical representations in {L}ogspace.
\newblock {\em SIAM J. on Computing}, 40(5):1292--1315, 2011.

\bibitem{KoeblerKV12}
J.~K{\"o}bler, S.~Kuhnert, and O.~Verbitsky.
\newblock Solving the canonical representation and {S}tar {S}ystem problems for
  proper circular-arc graphs in logspace.
\newblock In {\em Proc. Foundations of Software Technology and Theoretical
  Computer Science (FSTTCS)},
  number~18 in LIPIcs, pages 387--399. Leibniz-Zentrum
  für Informatik, 2012.

\bibitem{KoeblerKW12}
J.~K{\"o}bler, S.~Kuhnert, and O.~Watanabe.
\newblock Interval graph representation with given interval and intersection
  lengths.
\newblock {\em Algorithms
  and Computation -- Proc. 23rd ISAAC}, number 7676 in LNCS, pages 517--526. Springer, 2012.

\bibitem{KoeblerKV13}
J.~Köbler, S.~Kuhnert, and O.~Verbitsky.
\newblock Helly circular-arc graph isomorphism is in logspace.
\newblock In {\em Proc. 38th MFCS}, number
  8087 in LNCS, pages 631--642. Springer, 2013.

\bibitem{LinSS08}
M.~C. Lin, F.~J. Soulignac, and J.~L. Szwarcfiter.
\newblock A simple linear time algorithm for the isomorphism problem on proper
  circular-arc graphs.
\newblock In {\em Proc. 11th Scandinavian Workshop on Algorithm Theory (SWAT)},
number~5124 in LNCS, pages 355--366. Springer, 2008.

\bibitem{LuekerB79}
G.~Lueker and K.~Booth.
\newblock A linear time algorithm for deciding interval graph isomorphism.
\newblock {\em J. ACM}, 26(2):183--195, 1979.

\bibitem{McC03}
R.~M. McConnell.
\newblock Linear-time recognition of circular-arc graphs.
\newblock {\em Algorithmica}, 37(2):93--147, 2003.

\bibitem{McMc99}
T.~A. McKee and  F.~R. McMorris.
\newblock {\em Topics in intersection graph theory}.
\newblock SIAM Monographs on Discrete Mathematics and Applications 2. 
Philadelphia: SIAM, 1999.

\bibitem{OpsutR81}
R.~J. Opsut and F.~S. Roberts.
\newblock On the fleet maintenance, mobile radio frequency, task assignment,
  and traffic phasing problems.
\newblock In {\em The theory and applications of graphs}, pages 479--492.
  Wiley, 1981.

\bibitem{Sou14}
F.~J. Soulignac.
\newblock Minimal and short representations of unit interval and unit
  circular-arc graphs.
\newblock E-print: \url{http://arxiv.org/abs/1408.3443v2}, 2014.

\bibitem{Spinrad}
J.~Spinrad.
\newblock {\em Efficient graph representations}.
\newblock Number~19 in Field Institute Monographs. AMS, 2003.

\bibitem{Toran04}
J.~Torán.
\newblock
 On the hardness of Graph Isomorphism. 
\newblock
{\em SIAM J. Comput.} 33(5):1093--1108, 2004.

\bibitem{Uehara08}
R. Uehara.
\newblock
Simple geometrical intersection graphs.
\newblock 
In {\em Proc. 2nd Int. Workshop on Algorithms and Computation (WALCOM)},
number~4921 in LNCS, pages 25--33. Springer, 2008.

\bibitem{Ueh13}
R.~Uehara.
\newblock Tractabilities and intractabilities on geometric intersection graphs.
\newblock {\em Algorithms}, 6(1):60--83, 1 2013.

\end{thebibliography}

\end{document}
